\pdfoutput=1
\documentclass[11pt]{article}

\usepackage[margin=1in]{geometry}
\usepackage{amsmath,amssymb,amsthm,mathtools,bm,booktabs,array}
\usepackage{graphicx}
\usepackage{placeins}
\usepackage{microtype}
\usepackage{url}
\usepackage[colorlinks=true,linkcolor=blue,citecolor=blue,urlcolor=blue]{hyperref}

\allowdisplaybreaks
\newcommand{\R}{\mathbb{R}}
\newcommand{\Pbb}{\mathbb{P}}

\newcommand{\tr}{\operatorname{tr}}
\newcommand{\Cov}{\operatorname{Cov}}
\newcommand{\Var}{\operatorname{Var}}
\newcommand{\op}{\operatorname{op}}

\newcommand{\bmu}{\bm\mu}
\newcommand{\bSigma}{\mathbf{\Sigma}}
\newcommand{\bI}{\mathbf I}
\newcommand{\bS}{\mathbf S}
\newcommand{\bD}{\mathbf D}
\newcommand{\bR}{\mathbf R}
\newcommand{\bXmat}{\mathbf X}
\newcommand{\bZmat}{\mathbf Z}
\newcommand{\bOmega}{\mathbf{\Omega}}
\newcommand{\bA}{\mathbf A}
\newcommand{\bB}{\mathbf B}
\newcommand{\bDelta}{\bm\Delta}
\newcommand{\bdelta}{\bm\delta}

\newcommand{\calS}{\mathcal S}
\newcommand{\CCT}{\operatorname{CCT}}

\theoremstyle{plain}
\newtheorem{theorem}{Theorem}
\newtheorem{lemma}{Lemma}
\theoremstyle{definition}
\newtheorem{assumption}{Assumption}
\theoremstyle{remark}
\newtheorem{remark}{Remark}

\title{Cauchy Aggregation of Ridge-Regularized Hotelling Tests for High-Dimensional Change-Point Detection}

\author{%
Ping Zhao\\
\small School of Mathematical Science, Tiangong University, Tianjin, China\\
\small \texttt{zhaoping@tiangong.edu.cn}
\and
Le Zhou\\
\small Department of Mathematics, Hong Kong Baptist University, Hong Kong, China\\
\small \texttt{lezhou@hkbu.edu.hk}
\and
Long Feng\\
\small School of Statistics and Data Science, Nankai University, Tianjin, China\\
\small \texttt{flnankai@nankai.edu.cn}
}

\date{}

\begin{document}

\maketitle

\begin{abstract}
Ridge-regularized Hotelling-type (RHT) change-point tests depend on a ridge parameter $\lambda$, but the power-optimal value is determined by the unknown covariance structure and the unknown mean shift. We avoid selecting a single ridge value by computing fixed-ridge p-values on a finite deterministic grid and aggregating them with the Cauchy combination rule. Under the standard random-matrix conditions for fixed-ridge RHT statistics, we establish finite-grid joint weak convergence of the ridge processes. This leads to fixed-level validity under joint-limit calibration and small-tail validity for the analytic Cauchy p-value. Monte Carlo experiments show that deterministic-grid Cauchy aggregation has stable size behavior and achieves power close to the best stable fixed ridge choice across a range of covariance and signal configurations.
\end{abstract}



\bigskip
\section{Introduction}

Ridge regularization has long been used to stabilize inverse problems and covariance-based multivariate procedures.  In the classical low-dimensional setting, Tikhonov regularization and ridge regression add a positive multiple of the identity matrix to an otherwise unstable inverse; see Tikhonov \cite{tikhonov1963regularization} and Hoerl and Kennard \cite{hoerl1970ridge}.  The associated tuning-parameter literature was primarily estimation-oriented.  Cross-validation, generalized cross-validation and related risk criteria select the regularization level by balancing bias and variance; see Craven and Wahba \cite{craven1979smoothing}.  Similar shrinkage ideas were later used in regularized discriminant analysis and covariance estimation, where the tuning parameter interpolates between an empirical covariance matrix and a stable target; see Friedman \cite{friedman1989rda} and Ledoit and Wolf \cite{ledoit2004wellconditioned}.

In high-dimensional testing, ridge regularization plays a more structural role.  When $p$ is comparable to $n$, the sample covariance matrix is singular or ill-conditioned, and Hotelling-type tests based on $\bS_n^{-1}$ are no longer available without strong structural assumptions on $\bSigma_p$.  Ridge-regularized Hotelling statistics replace $\bS_n^{-1}$ by $(\bS_n+\lambda\bI_p)^{-1}$, producing a stable covariance normalization while retaining spectral information about the population covariance matrix.  This idea has been developed for high-dimensional mean testing by Chen et al. \cite{chen2011regularized} and Li et al. \cite{li2020adaptable}, and it is closely connected to deterministic-equivalent arguments from random matrix theory; see Bai and Silverstein \cite{bai2010spectral}.  For high-dimensional change-point inference, self-normalized and covariance-free methods have also been proposed; see Zhang and Lavitas \cite{zhang2018self} and Wang et al. \cite{wang2022inference}.  More recently, Li and Xu \cite{lixu2026ridgecp} introduced ridge-regularized CUSUM statistics for high-dimensional mean change-point detection and showed that, for each fixed deterministic $\lambda>0$, the resulting RHT scan statistic has a pivotal Gaussian-process null limit.

The remaining difficulty is the choice of $\lambda$.  A large ridge value moves the statistic toward an unnormalized $\ell_2$ CUSUM statistic, whereas a small ridge value uses covariance information more aggressively.  Power calculations for RHT-type tests show that the optimal value depends on the covariance structure of the mean shift and on its alignment with the eigenspaces of $\bSigma_p$.  Such information is unavailable in a change-point problem, where the location, direction, sparsity pattern and covariance alignment of the jump are all unknown under the alternative.  Consequently, choosing a single ridge value can lead to substantial power loss under misspecified alternatives.

This paper takes an aggregation approach.  We fix a finite deterministic grid $\Lambda_n=\{\lambda_{1,n},\ldots,\lambda_{L,n}\}$, compute the fixed-ridge RHT p-value at each grid point, and combine the p-values by the Cauchy combination rule of Liu and Xie \cite{liu2020cauchy}.  The method avoids estimating a power-optimal ridge value and is directly compatible with existing fixed-ridge RHT implementations.  On the theoretical side, finite-grid aggregation requires a joint null result across ridge values rather than only marginal fixed-$\lambda$ limits.  We derive this finite-grid joint convergence under the standard RHT random-matrix assumptions.  The marginal limiting processes remain pivotal, while the cross-ridge covariance may depend on the limiting spectral structure.  This distinction yields two calibration statements: joint-limit calibration gives asymptotically exact fixed-level size, and the analytic Cauchy p-value gives the usual small-tail validity without estimating cross-ridge dependence.  Simulations indicate that the proposed aggregation maintains reasonable size and tracks the best stable fixed ridge choice across covariance and signal designs.

The remainder of the paper is organized as follows.  Section~2 formulates the change-point testing problem, reviews the fixed-ridge RHT scan statistic from Li and Xu \cite{lixu2026ridgecp}, and introduces the deterministic-grid Cauchy aggregation rule.  Section~3 establishes the finite-grid joint null limit, size validity and conditional adaptive consistency.  Section~4 reports size and power simulations, and Section~5 concludes with a brief discussion.  Technical proofs are collected in the Appendix.

\section{Problem formulation, fixed-ridge statistics and Cauchy aggregation}

\subsection{Problem formulation and fixed-ridge statistics}

The model, scan statistic and fixed-ridge marginal calibration in this subsection follow Li and Xu \cite{lixu2026ridgecp}.  We restate them in a self-contained form because the proposed Cauchy aggregation starts from their fixed-ridge RHT p-values.  Let
\[
    \bm X_j=\bmu_j+\bSigma_p^{1/2}\bm Z_j,\qquad j=1,\ldots,n,
\]
where $\bm X_j\in\R^p$, $\bmu_j\in\R^p$, the covariance matrix $\bSigma_p$ is common over time, and the entries of $\bm Z_j$ are independent with mean zero and variance one.  Fix a trimming constant $\varepsilon\in(0,1/2)$.  The null hypothesis is
\[
    H_0:\quad \bmu_1=\cdots=\bmu_n=\bmu,
\]
for some unknown constant vector $\bmu\in\R^p$.  The alternative is that the mean sequence is piecewise constant with at least one abrupt change:
\[
    H_1:\quad \exists\ m\ge1,
    \quad 0=\tau_0<\tau_1<\cdots<\tau_m<\tau_{m+1}=1,
    \qquad
    \min_{0\le r\le m}(\tau_{r+1}-\tau_r)\ge\varepsilon,
\]
with vectors $\bmu^{(1)},\ldots,\bmu^{(m+1)}$ such that
\[
    \bmu_j=\bmu^{(r)},
    \qquad
    \lfloor n\tau_{r-1}\rfloor<j\le \lfloor n\tau_r\rfloor,
    \qquad r=1,\ldots,m+1,
\]
where at least one adjacent pair satisfies $\bmu^{(r)}\ne\bmu^{(r+1)}$.  In the single-change case, this reduces to
\[
    H_{1,\rm sc}:\quad
    \bmu_j=\bmu^{(1)}1\{j\le \lfloor n\tau_\star\rfloor\}
    +\bmu^{(2)}1\{j>\lfloor n\tau_\star\rfloor\},
    \qquad \bdelta_p=\bmu^{(2)}-\bmu^{(1)}\ne0,
\]
for some unknown $\tau_\star\in[\varepsilon,1-\varepsilon]$.  Thus the covariance matrix is not allowed to change; the testing problem is to detect changes in the high-dimensional mean sequence.

Let
\[
    \bS_n=n^{-1}\sum_{j=1}^n(\bm X_j-\bar{\bm X})(\bm X_j-\bar{\bm X})^T,
    \qquad \bar{\bm X}=n^{-1}\sum_{j=1}^n\bm X_j.
\]
For an ordered triple $s=(t_1,t_2,t_3)$ with $0\le t_1<t_2<t_3\le 1$, put $k(t)=\lfloor nt\rfloor+1$ and define the adjacent-segment contrast
\[
    \bDelta(s)=\frac{1}{k(t_3)-k(t_2)}\sum_{j=k(t_2)}^{k(t_3)-1}\bm X_j
    -\frac{1}{k(t_2)-k(t_1)}\sum_{j=k(t_1)}^{k(t_2)-1}\bm X_j.
\]
The effective sample size is
\[
    N(s)=\frac{\{k(t_2)-k(t_1)\}\{k(t_3)-k(t_2)\}}{k(t_3)-k(t_1)}.
\]
For a fixed deterministic ridge parameter $\lambda>0$, define
\[
    V_\lambda(s)=N(s)\bDelta(s)^T(\bS_n+\lambda\bI_p)^{-1}\bDelta(s).
\]
Let $\gamma_n=p/(n-1)$, and let $\alpha_1,\ldots,\alpha_p$ be the eigenvalues of $\bS_n$.
Set
\[
    \widehat m_\lambda=p^{-1}\sum_{j=1}^p\left\{\frac{n}{n-1}\alpha_j+\lambda\right\}^{-1},\quad
    \widehat m'_{\lambda}=p^{-1}\sum_{j=1}^p\left\{\frac{n}{n-1}\alpha_j+\lambda\right\}^{-2},
\]
\[
    \widehat\Theta_\lambda=1-\lambda\widehat m_\lambda,
\]
\[
    \widehat\Gamma_\lambda=2(1-\gamma_n+\gamma_n\lambda\widehat m_\lambda)(1-\lambda\widehat m_\lambda)
    -2(\lambda\widehat m_\lambda-\lambda^2\widehat m'_{\lambda}).
\]
The standardized local statistic is
\[
    D_\lambda(s)=\sqrt p\,\frac{p^{-1}V_\lambda(s)-\widehat\Theta_\lambda}{\widehat\Gamma_\lambda^{1/2}}.
\]

For a single change-point alternative we use
\[
    \calS_{\rm sc}(\varepsilon)=\{(0,t,1):\varepsilon\le t\le 1-\varepsilon\},
    \qquad T_{\lambda,{\rm sc}}=\sup_{s\in\calS_{\rm sc}(\varepsilon)}D_\lambda(s).
\]
For multiple changes one may use either the continuous set
\[
    \calS_{\rm mc}(\varepsilon)=\{(t_1,t_2,t_3):t_2-t_1\ge\varepsilon,
    \ t_3-t_2\ge\varepsilon\}
\]
or a deterministic discretized version.  We write $\calS$ for the selected scan set and
\[
    T_\lambda=\sup_{s\in\calS}D_\lambda(s),
\]
where $\calS=\calS_{\rm sc}(\varepsilon)$ in the single-change case and $\calS\subseteq\calS_{\rm mc}(\varepsilon)$ in the multiple-change case.  The theoretical result below is proved for the continuous single-change scan and for any fixed deterministic discretization of the multiple-change scan.  This is the version used in the simulations and avoids imposing additional entropy conditions on the three-dimensional continuous scan.

Under standard random-matrix assumptions, for each fixed deterministic $\lambda>0$,
\[
    T_\lambda \xrightarrow{d} T_0=\sup_{s\in\calS}G(s),
\]
where $G$ is a centered Gaussian process whose covariance is determined only by the overlap of the segment-contrast functions.
For the single-change scan, if $s_t=(0,t,1)$, then
\[
    \operatorname{cov}\{G(s_t),G(s_u)\}
    =\frac{\min(t,u)\{1-\max(t,u)\}}{\max(t,u)\{1-\min(t,u)\}}.
\]
Thus the marginal p-value
\[
    P_\lambda=1-F_{\calS}(T_\lambda),
\]
where $F_{\calS}$ is the distribution function of $T_0$, is asymptotically valid and is free of $\bSigma_p$.
The distribution $F_{\calS}$ is continuous for the nondegenerate Gaussian suprema considered here; for a finite scan this follows from nonsingularity of the Gaussian vector, and for the single-change scan it follows from standard anti-concentration results for Gaussian processes.  In practice $F_{\calS}$ is simulated once from the limiting Gaussian process.

\subsection{Cauchy combination over a deterministic ridge grid}

The fixed-ridge RHT statistic provides a valid test once $\lambda$ is fixed, but the ridge value that maximizes power is not universal.  A large ridge value moves the statistic toward an unnormalized $\ell_2$ CUSUM statistic, whereas a small ridge value uses covariance information more aggressively.  The favorable choice depends on the spectrum of $\bSigma_p$, the location of the change, and the alignment of the unknown mean shift with the eigenspaces of $\bSigma_p$.  These quantities are unavailable under $H_1$.  We therefore avoid selecting a single ridge parameter and instead aggregate several fixed-ridge p-values.  The Cauchy rule is used because its heavy-tailed transformation makes the combined statistic sensitive to the most significant ridge value while avoiding estimation of the cross-ridge dependence matrix.

Fix a finite deterministic grid
\[
    \Lambda_n=\{\lambda_{1,n},\ldots,\lambda_{L,n}\}\subset(0,\infty).
\]
Throughout the theory, $L$ is fixed and the grid is bounded away from zero and infinity.
For notational simplicity we write $\lambda_\ell$ for $\lambda_{\ell,n}$.
A convenient implementation after trace standardization is
\[
    \lambda_\ell=a_\ell\gamma_n,
    \qquad 0<a_1<\cdots<a_L<\infty,
\]
which is deterministic because $\gamma_n=p/(n-1)$ is fixed by the sample dimensions.
The simulations below use the deterministic grid $a_\ell\in\{0.05,0.10,\ldots,0.50\}$ for the size study.  The power plots display two representative fixed choices, $a=0.10$ and $a=0.20$, the Cauchy rule over a size-filtered deterministic subgrid, and the oracle best fixed ridge value over the same subgrid.

For each $\lambda_\ell$, compute $T_{\lambda_\ell}$ and its marginal p-value $P_{\lambda_\ell}$.
Let $w_\ell>0$ be fixed weights satisfying $\sum_{\ell=1}^L w_\ell=1$; equal weights are used in the simulations.
The Cauchy statistic and analytic Cauchy p-value are
\[
    C_n=\sum_{\ell=1}^L w_\ell\tan\{\pi(1/2-P_{\lambda_\ell})\},
    \qquad
    P_{\CCT}=\frac12-\frac{1}{\pi}\arctan(C_n).
\]
The analytic implementation rejects when $P_{\CCT}\le\alpha$, equivalently when $C_n\ge\cot(\pi\alpha)$.
This implementation does not estimate the dependence among the ridge-specific p-values.

For formal fixed-level calibration, define the critical value from the joint limit of the vector $(P_{\lambda_1},\ldots,P_{\lambda_L})$ and reject for large $C_n$.
This version has exact asymptotic size at any fixed level, but its joint critical value depends on the limiting cross-ridge dependence.
The analytic Cauchy p-value avoids this joint calibration and is justified by the usual CCT small-tail argument.
The two calibrations therefore have distinct roles: joint-limit calibration gives fixed-level exactness, while the analytic CCT gives a simple dependence-robust tail approximation.

\begin{remark}
The Cauchy rule is not a power-weighted average.
For small p-values, $\tan\{\pi(1/2-p)\}\sim(\pi p)^{-1}$, so the combined statistic is dominated by the most significant ridge value.
Hence the method behaves like a smooth adaptive minimum-p rule, but with a simple heavy-tailed reference distribution.
\end{remark}

\section{Asymptotic validity}

We state the theory for a fixed finite deterministic grid.  The regularity condition below is the same high-dimensional random-matrix condition used for the fixed-ridge RHT change-point theorem of Li and Xu \cite{lixu2026ridgecp}.  Under the null hypothesis the change-point-location part of their condition is vacuous.

\begin{assumption}[RHT regularity]\label{ass:rht}
The entries of $\bm Z_j$ are independent, have mean zero, unit variance and finite fourth moments.  Moreover $p/(n-1)\to\gamma\in(0,\infty)$, $\|\bSigma_p\|_{\op}$ is uniformly bounded, and the empirical spectral distribution $H_p(x)=p^{-1}\sum_{j=1}^p1\{\tau_{j,p}\le x\}$ of $\bSigma_p$, where $\tau_{1,p},\ldots,\tau_{p,p}$ are the eigenvalues of $\bSigma_p$, converges weakly to a nondegenerate compactly supported limit.  The trimming parameter $\varepsilon$ is fixed.  The ridge grid is deterministic and contains no data-dependent element:
\[
    \Lambda_n=\{\lambda_{1,n},\ldots,\lambda_{L,n}\},\qquad L<\infty.
\]
There are constants $0<\underline\lambda<\overline\lambda<\infty$ such that
\[
    \underline\lambda\le \min_{1\le \ell\le L}\lambda_{\ell,n}
    \le \max_{1\le \ell\le L}\lambda_{\ell,n}\le\overline\lambda,
\]
for all sufficiently large $n$, and $\lambda_{\ell,n}\to\lambda_{\ell,0}\in(0,\infty)$ for each $\ell$.  Repeated grid points are merged, so the limiting values $\lambda_{1,0},\ldots,\lambda_{L,0}$ are distinct.  The scanning set $\calS$ is either the single-change set $\calS_{\rm sc}(\varepsilon)$ or a deterministic discretized multiple-change set contained in $\calS_{\rm mc}(\varepsilon)$; the grid version used in the RHT paper is included.
\end{assumption}

For $s=(t_1,t_2,t_3)$, define the population contrast function
\[
    u_s(x)=
    \begin{cases}
    0, &0\le x<t_1,\\
    -(t_2-t_1)^{-1}, &t_1\le x<t_2,\\
    (t_3-t_2)^{-1}, &t_2\le x<t_3,\\
    0, &t_3\le x\le1,
    \end{cases}
\]
and write $\kappa(s,r)=\int_0^1u_s(x)u_r(x)\,dx$ and $\kappa(s)=\kappa(s,s)$.  Put
\[
    K_0(s,r)=\frac{\kappa(s,r)^2}{\kappa(s)\kappa(r)}.
\]
The next theorem derives the finite-grid joint null limit from the fixed-ridge RHT conditions.  The marginal covariance is pivotal, while the cross-ridge correlation is a deterministic spectral quantity.

\begin{theorem}[Finite-grid joint RHT null limit]\label{thm:joint}
Suppose Assumption~\ref{ass:rht} holds and $H_0$ is true.  Then there exist constants $\rho_{\ell m}\in[-1,1]$, $\rho_{\ell\ell}=1$, and a centered Gaussian vector process
\[
    \bm G^\Lambda(s)=\{G_1(s),\ldots,G_L(s)\}^T,
    \qquad s\in\calS,
\]
with almost surely continuous sample paths such that
\[
    \{D_{\lambda_{\ell,n}}(s):s\in\calS,
       \ell=1,\ldots,L\}
    \xrightarrow{d}
    \{G_\ell(s):s\in\calS,
       \ell=1,\ldots,L\}
\]
in $\ell^\infty(\calS)^L$.  Its covariance kernel is
\[
    \Cov\{G_\ell(s),G_m(r)\}=\rho_{\ell m}K_0(s,r),
    \qquad s,r\in\calS,
    \qquad 1\le \ell,m\le L.
\]
For each fixed $\ell$, $G_\ell$ has the usual fixed-ridge RHT null law.  In general $\rho_{\ell m}$, $\ell\ne m$, is determined by cross-products of deterministic equivalents of the two ridge resolvents and may depend on the limiting spectral distribution of $\bSigma_p$.  If $F_{\calS}$ denotes the continuous distribution function of $\sup_{s\in\calS}G_\ell(s)$, then
\[
    \bm P_n=(P_{\lambda_{1,n}},\ldots,P_{\lambda_{L,n}})^T
    \xrightarrow{d}
    \bm P_\infty=(P_{\infty,1},\ldots,P_{\infty,L})^T,
\]
where $P_{\infty,\ell}=1-F_{\calS}\{\sup_{s\in\calS}G_\ell(s)\}$.
\end{theorem}

Let
\[
    C(\bm p)=\sum_{\ell=1}^Lw_\ell\tan\{\pi(1/2-p_\ell)\},
    \qquad w_\ell>0,
    \qquad \sum_{\ell=1}^Lw_\ell=1.
\]
Let $c_{\alpha,L}$ be a continuity point satisfying
\[
    \Pbb\{C(\bm P_\infty)>c_{\alpha,L}\}=\alpha.
\]

\begin{theorem}[Size validity]\label{thm:size}
Under Assumption~\ref{ass:rht} and $H_0$,
\[
    \Pbb\{C(\bm P_n)>c_{\alpha,L}\}\to\alpha .
\]
For the analytic Cauchy p-value, for every fixed $\alpha\in(0,1/2)$ such that $\cot(\pi\alpha)$ is a continuity point of $C(\bm P_\infty)$,
\[
    \lim_{n\to\infty}\Pbb(P_{\CCT}\le\alpha)
    =\Pbb\{C(\bm P_\infty)\ge\cot(\pi\alpha)\}.
\]
Thus the analytic Cauchy p-value is exactly size-$\alpha$ at a fixed level only when the last probability equals $\alpha$.  If, in addition,
\[
    \Pbb\{C(\bm P_\infty)>t\}\sim \frac{1}{\pi t},
    \qquad t\to\infty,
\]
then
\[
    \lim_{\alpha\downarrow0}\lim_{n\to\infty}
    \frac{\Pbb(P_{\CCT}\le\alpha)}{\alpha}=1.
\]
\end{theorem}

Theorem~\ref{thm:size} is stated in two parts.  Joint-limit calibration gives fixed-level exactness.  The analytic Cauchy p-value gives the usual CCT small-tail approximation of Liu and Xie \cite{liu2020cauchy}; it does not by itself imply exact size at a fixed conventional level such as $0.05$ unless the fixed-level identity above holds.

Before stating the adaptive power result, we spell out the quantities used under $H_1$.  Let $H_p$ be the empirical spectral distribution in Assumption~\ref{ass:rht}, and let $\phi_n(z)=\phi(z;\gamma_n,H_p)$ denote the solution of the Mar\v{c}enko--Pastur equation associated with $(\gamma_n,H_p)$.  On the negative real axis define
\[
    \Theta_n(\lambda)=1-\lambda\phi_n(-\lambda),
\]
\[
    \Gamma_n(\lambda)=2\{1-\gamma_n+\gamma_n\lambda\phi_n(-\lambda)\}
    \{1-\lambda\phi_n(-\lambda)\}
    -2\{\lambda\phi_n(-\lambda)-\lambda^2\phi'_n(-\lambda)\},
\]
and
\[
    \bD_p(\lambda)=
    \left[\{1-\gamma_n+\gamma_n\lambda\phi_n(-\lambda)\}\bSigma_p
    +\lambda\bI_p\right]^{-1}.
\]
The matrix $\bD_p(\lambda)$ is the deterministic equivalent of $(\bS_n+\lambda\bI_p)^{-1}$, and $\Gamma_n(\lambda)$ is the deterministic counterpart of $\widehat\Gamma_\lambda$.

We use the probability-alternative formulation of Li and Xu \cite{lixu2026ridgecp}.  In the single-change case, the mean jump is written as
\[
    \bdelta_p=p^{-3/4}\bB_p\bm w,
\]
where $\bB_p$ is a deterministic $p\times p$ matrix with bounded operator norm and $\bm w$ has independent mean-zero, unit-variance entries with uniformly controlled fourth moments.  The ridge-specific signal strength is
\[
    q_p(\lambda,\bB_p)=p^{-1}\tr\{\bD_p(\lambda)\bB_p\bB_p^T\}.
\]
If the true single change occurs at $\tau_\star\in(0,1)$ and $s=(0,t,1)$, the location-shape function is
\[
    h_{\rm sc}(t;\tau_\star)=
    \begin{cases}
    t(1-t)^{-1}(1-\tau_\star)^2, & 0\le t<\tau_\star,\\
    t^{-1}(1-t)\tau_\star^2, & \tau_\star\le t\le1.
    \end{cases}
\]
For multiple changes with true fractions $0=\tau_0<\tau_1<\cdots<\tau_m<\tau_{m+1}=1$, define $v_r(x)=1\{\tau_{r-1}\le x<\tau_r\}$ and
\[
    \psi_r(s)=\kappa(s)^{-1/2}\int_0^1u_s(x)v_r(x)\,dx,
    \qquad r=1,\ldots,m+1,
\]
where $u_s$ and $\kappa(s)$ are defined above.  Let $\bm\psi(s)=(\psi_1(s),\ldots,\psi_{m+1}(s))^T$ and collect the segment mean vectors in $\mathbf U=(\bmu^{(1)},\ldots,\bmu^{(m+1)})\in\R^{p\times(m+1)}$.  Under a separable probability alternative, $\mathbf U=p^{-3/4}\bB_p\mathbf W\bOmega_{\rm alt}^T$, where $\mathbf W\in\R^{p\times r}$ has independent mean-zero, unit-variance entries and $\bOmega_{\rm alt}\in\R^{(m+1)\times r}$ is a fixed matrix for some fixed $r$.  The corresponding location-shape function is
\[
    h_{\rm mc}(s)=\|\bOmega_{\rm alt}^T\bm\psi(s)\|_2^2.
\]
Thus, for either the single-change scan with $h(s)=h_{\rm sc}(t;\tau_\star)$ or the multiple-change scan with $h(s)=h_{\rm mc}(s)$, the deterministic drift of the $\ell$th fixed-ridge process has the form
\[
    R_{n,\ell}(s)=
    \frac{q_p(\lambda_{\ell,n},\bB_p)}{\gamma_n\Gamma_n(\lambda_{\ell,n})^{1/2}}h(s)+o(1),
\]
with the $o(1)$ term interpreted uniformly over the selected scan set in the local alternatives considered by Li and Xu \cite{lixu2026ridgecp}.  The next theorem states the aggregation consequence in a form that only requires the existence of such a drift expansion.

\begin{theorem}[Conditional adaptive consistency]\label{thm:consistency}
Let
\[
    Y_{n,\ell}=\tan\{\pi(1/2-P_{\lambda_{\ell,n}})\}.
\]
Suppose that, under a sequence of alternatives, there exist stochastic processes $Z_{n,\ell}$ and deterministic drifts $R_{n,\ell}$ satisfying
\[
    \max_{1\le\ell\le L}\sup_{s\in\calS}
    |D_{\lambda_{\ell,n}}(s)-Z_{n,\ell}(s)-R_{n,\ell}(s)|\le a_n,
    \qquad a_n\to0,
\]
and
\[
    \max_{1\le\ell\le L}\sup_{s\in\calS}|Z_{n,\ell}(s)|=O_p(1).
\]
If there exists a grid point $\ell_0$ such that
\[
    R_{n,\ell_0}^*:=\sup_{s\in\calS}R_{n,\ell_0}(s)\to\infty,
\]
and if the remaining Cauchy terms do not asymptotically cancel the diverging positive term,
\[
    \left[\sum_{\ell\ne\ell_0} w_\ell Y_{n,\ell}\right]_{-}
    =o_p\{w_{\ell_0}Y_{n,\ell_0}\},
    \qquad [x]_- = \max(-x,0),
\]
then $P_{\CCT}\to0$ in probability.  Consequently, the analytic Cauchy-combined RHT test has power tending to one.  The no-cancellation condition is automatic when $\sum_{\ell\ne\ell_0}w_\ell Y_{n,\ell}=O_p(1)$.
\end{theorem}

\section{Simulation}

We conduct Monte Carlo experiments under the model
\[
    \bm X_j=\bmu_j+\bSigma^{1/2}\bm Z_j,
\]
where every covariance matrix is scaled so that $\tr(\bSigma)=p$.  Four primary covariance structures are considered: (i) \textit{ID}, $\bSigma=\bI_p$; (ii) \textit{Toeplitz(.3)}, where $\bSigma$ is Toeplitz with $(i,j)$ entry $0.3^{|i-j|}$; (iii) \textit{Poly Decay}, where the raw eigenvalues are $\tau_j=0.01+(p-j+0.1)^2$; and (iv) \textit{Exp Decay}, where the raw eigenvalues are $\tau_j=\exp(-3j/p)$.  Two additional dependence designs are included: \textit{Toeplitz(.6)}, with $(i,j)$ entry $0.6^{|i-j|}$, and \textit{CS(.6)}, with all off-diagonal correlations equal to $0.6$.

The trimming parameter is $\varepsilon=0.1$ and the nominal level is $\alpha=0.05$.  Marginal p-values are computed from $100,000$ simulations of the limiting single-change Gaussian-process null distribution.  The deterministic ridge grid for the size study is
\[
    \lambda/\gamma_n\in\{0.05,0.10,0.15,0.20,0.25,0.30,0.35,0.40,0.45,0.50\}.
\]
Under $H_0$, the Cauchy statistic combines all ten fixed-ridge p-values with equal weights.  We take $n\in\{200,400\}$ and $p\in\{100,200,400\}$; each configuration is based on $1,000$ replications under either Gaussian innovations or standardized $t_5$ innovations.

\begin{table}[p]
\centering
\caption{Empirical sizes for Gaussian innovations at nominal level $5\%$.}
\label{tab:size-normal}
\scriptsize
\setlength{\tabcolsep}{2.1pt}
\renewcommand{\arraystretch}{0.72}
\resizebox{\textwidth}{!}{%
\begin{tabular}{lllrrrrrrrrrrr}
\toprule
 & & & \multicolumn{10}{c}{$\lambda/\gamma_n$} & \\
\cmidrule(lr){4-13}
Cov. & $n$ & $p$ & 0.05 & 0.10 & 0.15 & 0.20 & 0.25 & 0.30 & 0.35 & 0.40 & 0.45 & 0.50 & CCT \\
\midrule
ID & 200 & 100 & 4.9 & 5.2 & 5.0 & 4.7 & 4.7 & 4.6 & 4.6 & 4.5 & 4.6 & 4.5 & 5.0 \\
ID & 200 & 200 & 1.7 & 2.3 & 3.1 & 3.4 & 3.7 & 4.3 & 4.6 & 4.9 & 5.0 & 5.4 & 3.6 \\
ID & 200 & 400 & 12.5 & 5.3 & 4.3 & 4.4 & 4.3 & 4.3 & 4.6 & 4.8 & 4.8 & 5.0 & 5.1 \\
ID & 400 & 100 & 8.3 & 8.3 & 8.4 & 8.6 & 8.6 & 8.6 & 8.6 & 8.8 & 8.9 & 8.9 & 8.7 \\
ID & 400 & 200 & 4.6 & 4.5 & 4.4 & 4.5 & 4.7 & 4.7 & 4.7 & 4.8 & 4.4 & 4.4 & 4.8 \\
ID & 400 & 400 & 2.3 & 3.1 & 3.9 & 4.3 & 4.2 & 4.4 & 4.7 & 4.8 & 4.7 & 4.6 & 4.3 \\
\addlinespace[1pt]
Toep(.3) & 200 & 100 & 5.1 & 5.5 & 5.5 & 5.8 & 6.0 & 5.9 & 6.2 & 6.3 & 6.5 & 6.6 & 6.2 \\
Toep(.3) & 200 & 200 & 1.6 & 2.7 & 3.3 & 3.3 & 3.4 & 3.6 & 4.0 & 3.8 & 4.2 & 4.3 & 3.1 \\
Toep(.3) & 200 & 400 & 9.7 & 5.5 & 4.4 & 4.1 & 4.1 & 4.2 & 4.2 & 4.1 & 4.1 & 4.0 & 4.6 \\
Toep(.3) & 400 & 100 & 8.1 & 8.1 & 7.8 & 8.0 & 7.9 & 8.0 & 8.1 & 8.3 & 8.4 & 8.2 & 7.8 \\
Toep(.3) & 400 & 200 & 3.3 & 3.8 & 4.3 & 4.4 & 4.4 & 4.5 & 4.7 & 4.9 & 5.0 & 5.2 & 4.5 \\
Toep(.3) & 400 & 400 & 2.9 & 3.4 & 3.6 & 3.9 & 3.6 & 3.8 & 4.1 & 4.0 & 4.0 & 4.1 & 3.5 \\
\addlinespace[1pt]
Poly & 200 & 100 & 5.2 & 5.9 & 6.7 & 7.0 & 7.2 & 7.6 & 7.6 & 7.7 & 7.4 & 7.7 & 7.2 \\
Poly & 200 & 200 & 3.5 & 4.3 & 4.1 & 4.3 & 4.9 & 5.6 & 5.5 & 5.6 & 5.7 & 5.8 & 5.0 \\
Poly & 200 & 400 & 5.4 & 4.3 & 4.8 & 4.8 & 5.0 & 5.1 & 5.5 & 5.3 & 5.4 & 5.4 & 5.0 \\
Poly & 400 & 100 & 7.9 & 8.5 & 8.8 & 8.6 & 8.9 & 9.1 & 9.5 & 9.8 & 9.6 & 9.7 & 8.9 \\
Poly & 400 & 200 & 5.1 & 5.5 & 5.7 & 6.0 & 6.2 & 6.1 & 6.1 & 6.5 & 6.7 & 6.7 & 6.3 \\
Poly & 400 & 400 & 3.3 & 4.2 & 4.3 & 4.6 & 5.0 & 5.3 & 5.6 & 6.1 & 6.1 & 6.4 & 5.3 \\
\addlinespace[1pt]
Exp & 200 & 100 & 3.6 & 4.6 & 4.6 & 4.5 & 4.6 & 4.8 & 5.0 & 5.2 & 5.7 & 5.8 & 4.6 \\
Exp & 200 & 200 & 2.5 & 3.5 & 3.4 & 3.7 & 4.1 & 4.3 & 4.4 & 4.5 & 5.1 & 4.9 & 4.0 \\
Exp & 200 & 400 & 7.1 & 5.1 & 4.9 & 4.9 & 4.7 & 4.6 & 4.7 & 4.7 & 4.9 & 5.0 & 5.2 \\
Exp & 400 & 100 & 8.2 & 8.0 & 7.9 & 7.9 & 8.5 & 8.4 & 8.5 & 8.6 & 8.6 & 8.4 & 8.5 \\
Exp & 400 & 200 & 5.9 & 5.8 & 5.9 & 6.0 & 6.1 & 5.9 & 5.9 & 6.1 & 6.1 & 6.4 & 6.0 \\
Exp & 400 & 400 & 3.7 & 4.8 & 4.9 & 5.3 & 5.5 & 5.8 & 6.1 & 6.6 & 6.7 & 6.9 & 5.9 \\
\addlinespace[1pt]
Toep(.6) & 200 & 100 & 5.5 & 5.8 & 6.0 & 6.2 & 6.8 & 6.7 & 7.2 & 7.3 & 7.4 & 7.6 & 6.9 \\
Toep(.6) & 200 & 200 & 2.6 & 3.8 & 4.1 & 4.6 & 5.1 & 5.3 & 5.4 & 6.0 & 6.2 & 6.3 & 4.8 \\
Toep(.6) & 200 & 400 & 5.6 & 4.2 & 4.2 & 4.2 & 4.5 & 4.7 & 4.7 & 4.9 & 5.1 & 5.2 & 4.7 \\
Toep(.6) & 400 & 100 & 8.4 & 8.4 & 8.5 & 8.7 & 8.8 & 9.1 & 9.2 & 9.1 & 9.1 & 9.3 & 9.1 \\
Toep(.6) & 400 & 200 & 4.6 & 4.8 & 4.9 & 4.9 & 5.1 & 5.4 & 5.2 & 4.8 & 5.1 & 5.1 & 4.8 \\
Toep(.6) & 400 & 400 & 3.1 & 3.9 & 4.5 & 4.5 & 4.4 & 4.4 & 4.9 & 4.8 & 5.0 & 5.2 & 4.7 \\
\addlinespace[1pt]
CS(.6) & 200 & 100 & 5.0 & 5.1 & 5.6 & 6.2 & 6.5 & 6.2 & 6.2 & 6.2 & 6.3 & 6.7 & 6.3 \\
CS(.6) & 200 & 200 & 2.7 & 3.5 & 4.0 & 4.5 & 4.7 & 5.5 & 5.6 & 5.7 & 5.7 & 6.0 & 4.9 \\
CS(.6) & 200 & 400 & 4.9 & 4.3 & 4.8 & 5.1 & 4.9 & 4.9 & 5.2 & 5.3 & 5.3 & 5.6 & 5.2 \\
CS(.6) & 400 & 100 & 6.8 & 6.7 & 6.8 & 6.9 & 7.0 & 7.0 & 6.9 & 7.1 & 7.2 & 7.0 & 7.0 \\
CS(.6) & 400 & 200 & 6.4 & 6.6 & 6.7 & 6.7 & 6.8 & 6.7 & 6.9 & 6.8 & 6.6 & 6.7 & 6.8 \\
CS(.6) & 400 & 400 & 3.4 & 4.3 & 4.8 & 4.6 & 4.7 & 4.6 & 4.8 & 5.2 & 5.1 & 5.3 & 4.7 \\
\bottomrule
\end{tabular}%
}
\vspace{-0.5em}\begin{flushleft}\tiny ID: identity; Toep(.3): Toeplitz with $(i,j)$ entry $0.3^{|i-j|}$; Poly: polynomial eigenvalue decay; Exp: exponential eigenvalue decay; Toep(.6): Toeplitz with $(i,j)$ entry $0.6^{|i-j|}$; CS(.6): compound symmetry with off-diagonal correlation $0.6$. Values are empirical rejection rates in percent.\end{flushleft}
\end{table}

\begin{table}[p]
\centering
\caption{Empirical sizes for standardized $t_5$ innovations at nominal level $5\%$.}
\label{tab:size-t}
\scriptsize
\setlength{\tabcolsep}{2.1pt}
\renewcommand{\arraystretch}{0.72}
\resizebox{\textwidth}{!}{%
\begin{tabular}{lllrrrrrrrrrrr}
\toprule
 & & & \multicolumn{10}{c}{$\lambda/\gamma_n$} & \\
\cmidrule(lr){4-13}
Cov. & $n$ & $p$ & 0.05 & 0.10 & 0.15 & 0.20 & 0.25 & 0.30 & 0.35 & 0.40 & 0.45 & 0.50 & CCT \\
\midrule
ID & 200 & 100 & 4.9 & 4.9 & 4.8 & 4.7 & 4.7 & 4.9 & 4.9 & 5.1 & 5.0 & 5.0 & 4.8 \\
ID & 200 & 200 & 1.7 & 2.7 & 3.3 & 3.6 & 3.8 & 4.3 & 4.4 & 4.8 & 5.0 & 5.1 & 3.9 \\
ID & 200 & 400 & 12.0 & 6.7 & 5.4 & 4.6 & 4.7 & 4.8 & 4.8 & 4.7 & 4.8 & 4.9 & 5.9 \\
ID & 400 & 100 & 7.7 & 7.8 & 7.8 & 8.0 & 7.8 & 7.7 & 7.7 & 7.7 & 7.8 & 8.0 & 7.7 \\
ID & 400 & 200 & 4.0 & 4.1 & 4.1 & 4.4 & 4.7 & 4.6 & 4.6 & 4.9 & 5.1 & 5.3 & 4.7 \\
ID & 400 & 400 & 3.1 & 3.6 & 4.0 & 4.0 & 4.2 & 3.8 & 4.0 & 4.3 & 4.4 & 4.4 & 4.1 \\
\addlinespace[1pt]
Toep(.3) & 200 & 100 & 4.6 & 4.8 & 5.2 & 5.3 & 5.4 & 5.5 & 5.5 & 5.6 & 5.5 & 5.3 & 5.5 \\
Toep(.3) & 200 & 200 & 1.6 & 2.6 & 3.0 & 3.4 & 3.9 & 4.3 & 4.7 & 4.8 & 4.9 & 5.0 & 3.8 \\
Toep(.3) & 200 & 400 & 10.6 & 5.7 & 4.4 & 4.0 & 4.1 & 4.4 & 4.5 & 4.8 & 4.8 & 5.0 & 4.5 \\
Toep(.3) & 400 & 100 & 7.4 & 7.4 & 7.4 & 7.3 & 7.2 & 7.3 & 7.5 & 7.6 & 7.6 & 7.5 & 7.2 \\
Toep(.3) & 400 & 200 & 4.6 & 4.7 & 4.7 & 4.8 & 4.6 & 5.1 & 5.2 & 5.4 & 5.7 & 5.7 & 4.9 \\
Toep(.3) & 400 & 400 & 3.1 & 3.6 & 3.4 & 3.7 & 3.8 & 3.7 & 3.6 & 3.8 & 4.0 & 4.1 & 3.9 \\
\addlinespace[1pt]
Poly & 200 & 100 & 5.3 & 5.2 & 5.3 & 5.3 & 5.6 & 5.5 & 5.6 & 6.1 & 6.4 & 6.4 & 5.8 \\
Poly & 200 & 200 & 3.3 & 4.0 & 4.5 & 5.0 & 5.5 & 5.8 & 6.0 & 6.0 & 6.1 & 6.1 & 5.4 \\
Poly & 200 & 400 & 4.2 & 3.6 & 3.5 & 3.6 & 3.8 & 3.8 & 4.1 & 4.0 & 4.1 & 4.2 & 3.7 \\
Poly & 400 & 100 & 9.3 & 8.9 & 9.4 & 9.1 & 9.3 & 9.7 & 9.6 & 9.5 & 9.2 & 9.2 & 9.7 \\
Poly & 400 & 200 & 4.7 & 5.2 & 5.4 & 5.5 & 5.5 & 5.2 & 5.7 & 6.1 & 6.1 & 6.2 & 5.2 \\
Poly & 400 & 400 & 3.9 & 4.0 & 3.8 & 4.3 & 4.4 & 4.3 & 4.5 & 4.7 & 5.0 & 4.8 & 4.1 \\
\addlinespace[1pt]
Exp & 200 & 100 & 5.1 & 5.4 & 5.6 & 5.8 & 5.9 & 6.4 & 6.4 & 6.3 & 6.7 & 6.9 & 6.1 \\
Exp & 200 & 200 & 3.4 & 4.9 & 5.6 & 6.0 & 6.0 & 6.1 & 6.3 & 6.5 & 6.7 & 6.8 & 6.2 \\
Exp & 200 & 400 & 6.4 & 4.8 & 4.5 & 4.5 & 4.7 & 4.9 & 4.9 & 5.0 & 4.9 & 4.9 & 5.0 \\
Exp & 400 & 100 & 8.2 & 8.4 & 8.4 & 8.1 & 8.1 & 8.2 & 8.4 & 8.8 & 9.1 & 9.0 & 8.2 \\
Exp & 400 & 200 & 6.1 & 6.3 & 6.4 & 6.3 & 6.5 & 6.3 & 6.2 & 6.2 & 6.6 & 6.8 & 6.5 \\
Exp & 400 & 400 & 3.2 & 3.8 & 4.3 & 4.4 & 4.5 & 4.7 & 4.7 & 4.9 & 4.8 & 4.8 & 4.2 \\
\addlinespace[1pt]
Toep(.6) & 200 & 100 & 5.6 & 5.7 & 5.6 & 5.7 & 5.9 & 6.4 & 6.8 & 6.8 & 6.7 & 6.6 & 6.2 \\
Toep(.6) & 200 & 200 & 3.0 & 3.9 & 4.4 & 4.7 & 5.2 & 5.5 & 5.9 & 6.2 & 6.4 & 6.5 & 5.3 \\
Toep(.6) & 200 & 400 & 6.4 & 5.1 & 5.2 & 5.4 & 5.4 & 5.7 & 5.9 & 5.7 & 5.8 & 6.1 & 5.9 \\
Toep(.6) & 400 & 100 & 8.4 & 8.4 & 8.5 & 8.5 & 8.6 & 8.5 & 8.3 & 8.4 & 8.5 & 8.4 & 8.8 \\
Toep(.6) & 400 & 200 & 5.6 & 5.6 & 5.8 & 5.9 & 6.2 & 6.3 & 6.4 & 6.6 & 7.0 & 7.2 & 6.3 \\
Toep(.6) & 400 & 400 & 2.4 & 3.2 & 3.3 & 3.6 & 4.5 & 4.5 & 4.9 & 5.2 & 5.5 & 5.7 & 4.3 \\
\addlinespace[1pt]
CS(.6) & 200 & 100 & 4.5 & 4.2 & 4.7 & 4.7 & 5.1 & 5.2 & 5.5 & 5.7 & 5.8 & 5.8 & 5.1 \\
CS(.6) & 200 & 200 & 3.0 & 3.9 & 4.0 & 4.1 & 4.8 & 4.6 & 4.7 & 5.0 & 4.9 & 5.1 & 4.6 \\
CS(.6) & 200 & 400 & 5.6 & 4.6 & 4.5 & 4.1 & 4.4 & 4.7 & 4.8 & 5.2 & 5.3 & 5.4 & 4.7 \\
CS(.6) & 400 & 100 & 7.7 & 7.5 & 7.2 & 7.2 & 7.4 & 7.3 & 7.3 & 7.4 & 7.5 & 7.8 & 7.5 \\
CS(.6) & 400 & 200 & 4.7 & 4.7 & 4.9 & 4.7 & 4.5 & 4.9 & 4.6 & 4.8 & 4.9 & 4.8 & 5.1 \\
CS(.6) & 400 & 400 & 4.0 & 4.3 & 4.4 & 4.6 & 4.6 & 4.9 & 5.1 & 5.2 & 5.0 & 5.0 & 4.4 \\
\bottomrule
\end{tabular}%
}
\vspace{-0.5em}\begin{flushleft}\tiny ID: identity; Toep(.3): Toeplitz with $(i,j)$ entry $0.3^{|i-j|}$; Poly: polynomial eigenvalue decay; Exp: exponential eigenvalue decay; Toep(.6): Toeplitz with $(i,j)$ entry $0.6^{|i-j|}$; CS(.6): compound symmetry with off-diagonal correlation $0.6$. Values are empirical rejection rates in percent.\end{flushleft}
\end{table}

Tables~\ref{tab:size-normal} and \ref{tab:size-t} show that moderate ridge values usually give rejection rates close to the nominal level.  The smallest grid value can be liberal when $p/n$ is large, especially for $n=200$ and $p=400$, reflecting the finite-sample instability of very aggressive covariance normalization.  The analytic Cauchy combination inherits this behavior when the grid contains unstable small ridge values, but remains well controlled once the included ridge values have reasonable marginal size.  This pattern motivates the use of deterministic grids that exclude extremely small ridge parameters.

For the power study, we fix $n=200$ and $p=400$.  There is a single change at $k_0=\lfloor n/2\rfloor$.  Three signal models are considered: dense independent Gaussian shifts $\bdelta\sim N(0,c\bI_p)$; covariance-aligned Gaussian shifts $\bdelta\sim N(0,c\bSigma)$; and sparse shifts with three randomly selected nonzero coordinates equal to $\pm5c$.  The signal strength $c$ varies over a panel-specific six-point grid chosen to display the transition from size to high power.  Four procedures are reported: fixed $\lambda/\gamma_n=0.1$, fixed $\lambda/\gamma_n=0.2$, the analytic Cauchy rule over the size-filtered deterministic subgrid, and the oracle best single ridge value over that same subgrid.  The size filter retains, for each covariance model, those fixed ridge values whose empirical Gaussian size is at most $10\%$; this is used only as a simulation device for comparing power under approximately controlled size.  The power curves are based on $10,000$ replications.

\begin{figure}[p]
\centering
\includegraphics[width=0.98\textwidth]{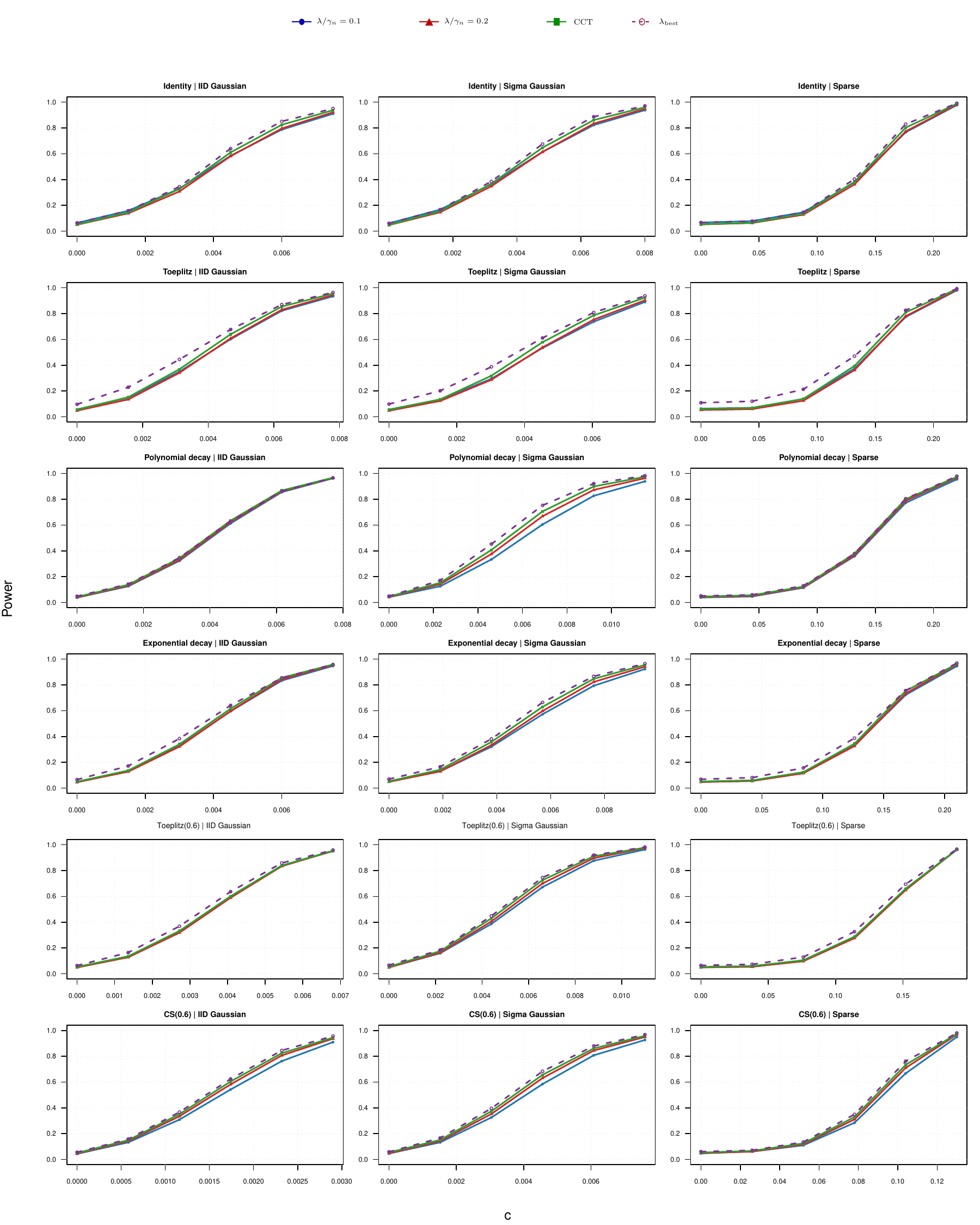}
\caption{Empirical power curves for $n=200$ and $p=400$.  Rows correspond to ID, Toeplitz(.3), Poly Decay, Exp Decay, Toeplitz(.6) and CS(.6); columns correspond to dense independent Gaussian, covariance-aligned Gaussian and sparse shifts.  In the legend, $\mathrm{CCT}$ denotes the Cauchy-combined test and $\lambda_{\mathrm{best}}$ denotes the oracle best fixed ridge value on the same size-filtered deterministic subgrid.}
\label{fig:power}
\end{figure}
\FloatBarrier

Figure~\ref{fig:power} shows that the Cauchy-combined test is generally close to the better of the two displayed fixed-ridge tests and follows the oracle best single ridge curve without requiring an oracle choice.  The advantage is most visible when the covariance structure changes the relative performance of $\lambda/\gamma_n=0.1$ and $0.2$.  For sparse shifts, dense RHT-type procedures are not designed to dominate max-type tests, but the Cauchy curve still tracks the better stable ridge choice.  In covariance-aligned Gaussian panels, the combined test is useful because the favorable ridge value varies with the spectrum of $\bSigma$.  Overall, the simulations support deterministic-grid aggregation: moderate ridge values control size, and Cauchy combination protects power against choosing a single suboptimal ridge.

To summarize the distance from the oracle curve in Figure~\ref{fig:power}, let $\mathcal C_{\bSigma,H}$ be the six signal strengths used for a fixed covariance design $\bSigma$ and signal class $H$.  For a method $m\in\{0.1,0.2,\mathrm{CCT}\}$, define
\[
    \Delta_m(\bSigma,H)=\sum_{c\in\mathcal C_{\bSigma,H}}
    \{P_{\mathrm{best}}(c)-P_m(c)\},
\]
where $P_{\mathrm{best}}(c)$ is the empirical power of the oracle best fixed ridge value on the size-filtered subgrid and $P_m(c)$ is the empirical power of method $m$.  Thus a smaller value of $\Delta_m(\bSigma,H)$ indicates a smaller cumulative loss relative to the oracle fixed-ridge benchmark.

\begin{table}[!htbp]
\centering
\caption{Cumulative power gap to the oracle best fixed ridge value over the six signal strengths in each panel of Figure~\ref{fig:power}.}
\label{tab:power-gap}
\scriptsize
\setlength{\tabcolsep}{3.2pt}
\renewcommand{\arraystretch}{0.86}
\resizebox{\textwidth}{!}{%
\begin{tabular}{lrrrrrrrrr}
\toprule
& \multicolumn{3}{c}{Dense independent Gaussian}
& \multicolumn{3}{c}{Covariance-aligned Gaussian}
& \multicolumn{3}{c}{Sparse} \\
\cmidrule(lr){2-4}\cmidrule(lr){5-7}\cmidrule(lr){8-10}
Covariance
& $\Delta_{0.1}$ & $\Delta_{0.2}$ & $\Delta_{\rm CCT}$
& $\Delta_{0.1}$ & $\Delta_{0.2}$ & $\Delta_{\rm CCT}$
& $\Delta_{0.1}$ & $\Delta_{0.2}$ & $\Delta_{\rm CCT}$ \\
\midrule
ID            & 0.1752 & 0.2138 & 0.1114 & 0.1800 & 0.2081 & 0.1061 & 0.1025 & 0.1561 & 0.0902 \\
Toeplitz(.3)  & 0.3671 & 0.3769 & 0.2551 & 0.3989 & 0.3938 & 0.2486 & 0.3319 & 0.3675 & 0.2696 \\
Poly Decay    & 0.0821 & 0.0496 & 0.0047 & 0.4592 & 0.2651 & 0.1437 & 0.1175 & 0.0689 & 0.0233 \\
Exp Decay     & 0.1957 & 0.1869 & 0.1194 & 0.3215 & 0.2390 & 0.1262 & 0.1955 & 0.1862 & 0.1233 \\
Toeplitz(.6)  & 0.1567 & 0.1811 & 0.1361 & 0.2411 & 0.1572 & 0.0728 & 0.1441 & 0.1659 & 0.1265 \\
CS(.6)        & 0.3071 & 0.1594 & 0.0821 & 0.3308 & 0.1848 & 0.0969 & 0.2387 & 0.1429 & 0.0777 \\
\bottomrule
\end{tabular}%
}
\vspace{-0.5em}\begin{flushleft}\tiny $\Delta_{0.1}$ and $\Delta_{0.2}$ correspond to the fixed choices $\lambda/\gamma_n=0.1$ and $0.2$; $\Delta_{\rm CCT}$ corresponds to the analytic Cauchy-combined test.\end{flushleft}
\end{table}
\FloatBarrier

Table~\ref{tab:power-gap} confirms the visual pattern in Figure~\ref{fig:power}.  The Cauchy-combined test has the smallest cumulative loss in every covariance-signal panel.  The reduction is especially clear under Poly Decay, CS(.6), and the covariance-aligned Gaussian alternatives, where the best fixed ridge value varies more strongly with the spectral structure.  This supports the main practical message: aggregating a stable deterministic ridge grid can recover most of the oracle fixed-ridge power without selecting a single value of $\lambda$.

\section{Discussion}

This paper studies ridge tuning in high-dimensional RHT change-point tests from an aggregation perspective.  Rather than estimating a single power-optimal ridge value, the proposed procedure evaluates fixed-ridge statistics on a finite deterministic grid and combines their marginal p-values by the Cauchy rule.  The main theoretical contribution is a finite-grid joint null limit for the ridge processes.  This result justifies fixed-level joint-limit calibration and clarifies the scope of the analytic Cauchy p-value, which provides a simple small-tail approximation without estimating cross-ridge dependence.

The simulations support the proposed use of deterministic-grid aggregation.  Moderate ridge values lead to stable marginal size, while the Cauchy statistic protects power against selecting a suboptimal fixed ridge value.  The results also show that very small ridge values can be unstable in finite samples when $p/n$ is large.  In applications, the deterministic grid should therefore be chosen away from zero unless additional calibration is used.

Two extensions are natural.  First, one could develop fixed-level calibration for the analytic Cauchy statistic by estimating the finite-grid joint limit or by constructing a valid bootstrap for the cross-ridge p-value vector.  Second, the deterministic finite-grid theory could be extended to growing or data-adaptive ridge grids, where uniform stochastic equicontinuity in $\lambda$ becomes essential.

\appendix
\setcounter{lemma}{0}
\setcounter{theorem}{0}
\setcounter{assumption}{0}
\setcounter{remark}{0}
\renewcommand{\thelemma}{\Alph{section}.\arabic{lemma}}
\renewcommand{\thetheorem}{\Alph{section}.\arabic{theorem}}
\renewcommand{\theassumption}{\Alph{section}.\arabic{assumption}}
\renewcommand{\theremark}{\Alph{section}.\arabic{remark}}

\section{Proofs}

The notation below is used only in the appendix.  Put
\[
    \bXmat=(\bm X_1,\ldots,\bm X_n),
    \qquad
    \bOmega_n=n^{-1}\bXmat\bXmat^T,
    \qquad
    \bR_\lambda=(\bOmega_n+\lambda\bI_p)^{-1}.
\]
Let $\phi_n(-\lambda)$ be the deterministic Marcenko--Pastur transform associated with $(\gamma_n,H_p)$ and define
\[
    a_n(\lambda)=1-\gamma_n+\gamma_n\lambda\phi_n(-\lambda),
    \qquad
    \Theta_n(\lambda)=1-\lambda\phi_n(-\lambda),
\]
\[
    \Gamma_n(\lambda)=2a_n(\lambda)\Theta_n(\lambda)
    -2\{\lambda\phi_n(-\lambda)-\lambda^2\phi'_n(-\lambda)\},
    \qquad
    \bD_n(\lambda)=\{a_n(\lambda)\bSigma_p+\lambda\bI_p\}^{-1}.
\]
The functions $a_n(\lambda)$, $\Theta_n(\lambda)$ and $\Gamma_n(\lambda)$ and the matrix $\bD_n(\lambda)$ are deterministic.  On $[\underline\lambda,\overline\lambda]$,
\[
    0<c\le a_n(\lambda)\le C,
    \qquad
    \|\bD_n(\lambda)\|_{\op}\le \underline\lambda^{-1},
    \qquad
    0<c\le \Gamma_n(\lambda)\le C<\infty,
\]
where the constants do not depend on $n$ or $\lambda$; these bounds are the negative-real-axis consequences of the Marcenko--Pastur equation used in the fixed-ridge RHT theorem \cite{bai2010spectral,lixu2026ridgecp}.

\begin{lemma}[Discrete contrast bounds]\label{lem:contrast}
For $s=(t_1,t_2,t_3)\in\calS$, define
\[
\tilde u_{n,j}(s)=\sqrt n
\begin{cases}
0, & j<k(t_1),\\
-\{k(t_2)-k(t_1)\}^{-1}, & k(t_1)\le j<k(t_2),\\
\{k(t_3)-k(t_2)\}^{-1}, & k(t_2)\le j<k(t_3),\\
0, & j\ge k(t_3).
\end{cases}
\]
Then, uniformly over $s,r\in\calS$,
\[
    n^{-1/2}\bXmat\tilde{\bm u}_n(s)=\bDelta(s),
\]
\[
    \|\tilde{\bm u}_n(s)\|_2^2=\kappa(s)+O(n^{-1}),
    \qquad
    \tilde{\bm u}_n(s)^T\tilde{\bm u}_n(r)=\kappa(s,r)+O(n^{-1}),
\]
\[
    \|\tilde{\bm u}_n(s)\|_2\le C_\varepsilon,
    \qquad
    \sqrt n\|\tilde{\bm u}_n(s)\|_\infty\le C_\varepsilon,
\]
and
\[
    \|\tilde{\bm u}_n(s)-\tilde{\bm u}_n(r)\|_2^2
    \le C_\varepsilon \|s-r\|_1+C_\varepsilon n^{-1}.
\]
\end{lemma}

\begin{proof}
For $m_{ab}=k(t_b)-k(t_a)$,
\[
    n^{-1/2}\bXmat\tilde{\bm u}_n(s)
    =m_{23}^{-1}\sum_{j=k(t_2)}^{k(t_3)-1}\bm X_j
     -m_{12}^{-1}\sum_{j=k(t_1)}^{k(t_2)-1}\bm X_j=
     \bDelta(s).
\]
Since $m_{12}\ge n\varepsilon-2$ and $m_{23}\ge n\varepsilon-2$,
\[
    \|\tilde{\bm u}_n(s)\|_2^2
    =n(m_{12}^{-1}+m_{23}^{-1})
    =\frac{1}{t_2-t_1}+\frac{1}{t_3-t_2}+O(n^{-1})
    =\kappa(s)+O(n^{-1}),
\]
\[
    \sqrt n\|\tilde{\bm u}_n(s)\|_\infty
    \le n\max(m_{12}^{-1},m_{23}^{-1})
    \le \frac{2}{\varepsilon}.
\]
For $I_{ab}(s)=\{j:k(t_a)\le j<k(t_b)\}$,
\[
\begin{aligned}
\tilde{\bm u}_n(s)^T\tilde{\bm u}_n(r)
&=n\sum_{a,b\in\{1,2\}}
   c_a(s)c_b(r)|I_a(s)\cap I_b(r)|,\\
|n^{-1}|I_a(s)\cap I_b(r)|
&\hspace{1.5em}
-\int_0^1 1_{I_a(s)}(x)1_{I_b(r)}(x)\,dx|\le 4n^{-1},
\end{aligned}
\]
where $c_1(s)=-(t_2-t_1)^{-1}+O(n^{-1})$, $c_2(s)=(t_3-t_2)^{-1}+O(n^{-1})$ and all denominators are bounded below by $\varepsilon/2$.  Hence
\[
    \tilde{\bm u}_n(s)^T\tilde{\bm u}_n(r)=\int_0^1u_s(x)u_r(x)\,dx+O(n^{-1}).
\]
The map $(t_1,t_2,t_3)\mapsto u_s$ is Lipschitz in $L^2[0,1]$ on the trimmed set except for moving endpoints; the endpoint contribution is bounded by the length of the symmetric difference.  Therefore
\[
    \int_0^1\{u_s(x)-u_r(x)\}^2dx\le C_\varepsilon\|s-r\|_1,
\]
which together with the discretization error gives the displayed bound.
\end{proof}

\begin{lemma}[Finite-grid cross-resolvent deterministic equivalents]\label{lem:crossres}
Under Assumption~\ref{ass:rht}, for every pair $\lambda_{a,n},\lambda_{b,n}\in\Lambda_n$ there exist deterministic quantities
\[
    \mathcal K^{\Sigma\Sigma}_{ab,n},\qquad
    \mathcal K^{\Sigma}_{ab,n},\qquad
    \mathcal K^{I}_{ab,n},
\]
with finite limits such that
\[
\begin{aligned}
&p^{-1}\tr\{\bR_{\lambda_{a,n}}\bSigma_p\bR_{\lambda_{b,n}}\bSigma_p\}
      -\mathcal K^{\Sigma\Sigma}_{ab,n}=o_p(1),\\
&p^{-1}\tr\{\bR_{\lambda_{a,n}}\bR_{\lambda_{b,n}}\bSigma_p\}
      -\mathcal K^{\Sigma}_{ab,n}=o_p(1),\\
&p^{-1}\tr\{\bR_{\lambda_{a,n}}\bR_{\lambda_{b,n}}\}
      -\mathcal K^{I}_{ab,n}=o_p(1).
\end{aligned}
\]
For $\bA\in\{\bI_p,\bSigma_p\}$, define
\[
    \vartheta_{\bA,n}(\lambda)=p^{-1}\tr\{\bD_n(\lambda)\bA\}.
\]
Then the deterministic equivalents without a middle population covariance are
\[
\mathcal K^{\bA}_{ab,n}=
\begin{cases}
\displaystyle
\frac{\vartheta_{\bA,n}(\lambda_{a,n})-
      \vartheta_{\bA,n}(\lambda_{b,n})}
     {\lambda_{b,n}-\lambda_{a,n}},
     & \lambda_{a,n}\ne\lambda_{b,n},\\[2.0ex]
\displaystyle
-\partial_\lambda\vartheta_{\bA,n}(\lambda)
 \big|_{\lambda=\lambda_{a,n}},
     & \lambda_{a,n}=\lambda_{b,n},
\end{cases}
\]
where $\mathcal K^I_{ab,n}=\mathcal K^{\bA}_{ab,n}$ with $\bA=\bI_p$ and
$\mathcal K^\Sigma_{ab,n}=\mathcal K^{\bA}_{ab,n}$ with $\bA=\bSigma_p$.  The quantity
$\mathcal K^{\Sigma\Sigma}_{ab,n}$ is the two-sided deterministic equivalent for the product
$\bR_{\lambda_{a,n}}\bSigma_p\bR_{\lambda_{b,n}}\bSigma_p$.  It is not, in general, equal to
$p^{-1}\tr\{\bD_n(\lambda_{a,n})\bSigma_p\bD_n(\lambda_{b,n})\bSigma_p\}$.
\end{lemma}

\begin{proof}
For $\bA\in\{\bI_p,\bSigma_p\}$ and $\lambda\in[\underline\lambda,\overline\lambda]$, Lemma 3.1 of Li and Xu \cite{lixu2026ridgecp} gives
\[
    p^{-1}\tr\{\bR_\lambda\bA\}-\vartheta_{\bA,n}(\lambda)=O_p(n^{-1/2}).
\]
If $\lambda_{a,n}\ne\lambda_{b,n}$, then
\[
    \bR_{\lambda_{a,n}}\bR_{\lambda_{b,n}}
    =\frac{\bR_{\lambda_{a,n}}-\bR_{\lambda_{b,n}}}
           {\lambda_{b,n}-\lambda_{a,n}},
\]
and hence
\[
\begin{aligned}
&p^{-1}\tr\{\bR_{\lambda_{a,n}}\bR_{\lambda_{b,n}}\bA\}-\mathcal K^{\bA}_{ab,n}\\
&\quad=\frac{
 p^{-1}\tr\{\bR_{\lambda_{a,n}}\bA\}-\vartheta_{\bA,n}(\lambda_{a,n})
-\big[p^{-1}\tr\{\bR_{\lambda_{b,n}}\bA\}-\vartheta_{\bA,n}(\lambda_{b,n})\big]}
 {\lambda_{b,n}-\lambda_{a,n}}.
\end{aligned}
\]
Since $\lambda_{a,n}\to\lambda_{a,0}$, $\lambda_{b,n}\to\lambda_{b,0}$ and the finite grid has distinct limiting points after duplicate values are removed, the denominator is bounded away from zero for $a\ne b$.  Thus the last display is $O_p(n^{-1/2})$.  If $\lambda_{a,n}=\lambda_{b,n}$, which can occur only for $a=b$ after repeated grid values have been merged, the preceding difference quotient is replaced by the two-resolvent deterministic equivalent at coincident spectral parameters.  Equivalently, for $\bA\in\{\bI_p,\bSigma_p\}$,
\[
    p^{-1}\tr\{\bR_\lambda^2\bA\}
    +\partial_\lambda \vartheta_{\bA,n}(\lambda)=o_p(1),
\]
with the derivative taken on the negative real axis.  This derivative form is the coincident-parameter limit of the deterministic-equivalent system; it is justified by the uniform resolvent bound
\[
    \|\bR_\lambda\|_{\op}\le \underline\lambda^{-1},
    \qquad
    \|\bR_{\lambda+h}-\bR_\lambda\|_{\op}\le |h|\underline\lambda^{-2},
\]
and the locally uniform deterministic-equivalent convergence on compact subsets of $(0,\infty)$.
The two-sided trace
\[
    p^{-1}\tr\{\bR_{\lambda_{a,n}}\bSigma_p\bR_{\lambda_{b,n}}\bSigma_p\}
\]
is a standard bilinear-resolvent deterministic equivalent for sample covariance matrices.  Under the separable covariance model $\bXmat=\bSigma_p^{1/2}\bZmat$, the negative-real-axis version follows from the deterministic-equivalent theory of Hachem, Loubaton and Najim \cite{hachem2007deterministic} and the bilinear resolvent bounds of Hachem et al. \cite{hachem2013bilinear}; the finite-fourth-moment reduction is obtained by the truncation argument used in Bai and Silverstein \cite{bai2004clt} and in Li and Xu \cite{lixu2026ridgecp}.  Therefore there is a deterministic $\mathcal K^{\Sigma\Sigma}_{ab,n}$ such that
\[
    p^{-1}\tr\{\bR_{\lambda_{a,n}}\bSigma_p\bR_{\lambda_{b,n}}\bSigma_p\}
      -\mathcal K^{\Sigma\Sigma}_{ab,n}=o_p(1).
\]
Because $\|\bSigma_p\|_{\op}\le C$ and the ridge parameters lie in a compact subset of $(0,\infty)$,
\[
    |\mathcal K^{\Sigma\Sigma}_{ab,n}|+|\mathcal K^{\Sigma}_{ab,n}|+|\mathcal K^{I}_{ab,n}|
    \le C.
\]
The weak convergence of $H_p$ and the convergence $\lambda_{a,n}\to\lambda_{a,0}$ imply convergence of $\mathcal K^{I}_{ab,n}$ and $\mathcal K^{\Sigma}_{ab,n}$ through the displayed formulas.  The deterministic-equivalent system in Hachem, Loubaton and Najim \cite{hachem2007deterministic} is continuous on the negative real axis, so $\mathcal K^{\Sigma\Sigma}_{ab,n}$ also has a finite limit.
\end{proof}

\begin{lemma}[Finite-grid regular-vector CLT]\label{lem:rvclt}
Let $\bm u_q,\bm v_q\in\R^n$, $q=1,\ldots,Q$, be deterministic regular vectors satisfying
\[
    \max_q(\|\bm u_q\|_2+\|\bm v_q\|_2)\le C,
    \qquad
    \sqrt n\max_q(\|\bm u_q\|_\infty+\|\bm v_q\|_\infty)\le C.
\]
Let $\lambda_{a_q,n}\in\Lambda_n$ and $b_q\in\R$ be fixed.  Define
\[
    H_{n,q}=p^{-1}\bm u_q^T\bXmat^T\bR_{\lambda_{a_q,n}}\bXmat\bm v_q,
    \qquad
    L_n=\sum_{q=1}^Q b_q\sqrt p\{H_{n,q}-EH_{n,q}\}.
\]
Then $L_n \xrightarrow{d} N(0,\sigma_L^2)$, where
\[
    \sigma_L^2
    =\lim_{n\to\infty}
      \sum_{q_1=1}^Q\sum_{q_2=1}^Q
      b_{q_1}b_{q_2}
      \mathcal V_{q_1q_2,n},
\]
\[
    \mathcal V_{q_1q_2,n}
    =\frac12\Gamma_{a_{q_1}a_{q_2},n}
      \{(\bm u_{q_1}^T\bm u_{q_2})(\bm v_{q_1}^T\bm v_{q_2})
       +(\bm u_{q_1}^T\bm v_{q_2})(\bm v_{q_1}^T\bm u_{q_2})\}
      +O(n^{-1/2}),
\]
where $\Gamma_{ab,n}$ is deterministic, $\Gamma_{aa,n}=\Gamma_n(\lambda_{a,n})$, and $\Gamma_{ab,n}\to\Gamma_{ab}$.  The numbers $\Gamma_{ab}$ are finite and are functions of the limiting cross-resolvent traces in Lemma~\ref{lem:crossres}.
\end{lemma}

\begin{proof}
Let $\mathcal F_j=\sigma(\bm X_1,\ldots,\bm X_j)$, $E_j(\cdot)=E(\cdot\mid\mathcal F_j)$, and
\[
    \xi_{n,j}=(E_j-E_{j-1})L_n,
    \qquad
    L_n=\sum_{j=1}^n\xi_{n,j}.
\]
Then
\[
    E_{j-1}\xi_{n,j}=0,
    \qquad
    E\xi_{n,j}=0.
\]
For $\bR_{j,\lambda}=(\bOmega_n^{(j)}+\lambda\bI_p)^{-1}$ and
\[
    \bOmega_n^{(j)}=n^{-1}\sum_{i\ne j}\bm X_i\bm X_i^T,
\]
$\bR_{j,\lambda}$ is independent of $\bm X_j$ but is not $\mathcal F_{j-1}$-measurable.  The proof uses only this independence.  Woodbury's identity gives
\[
    \bR_\lambda=\bR_{j,\lambda}
    -\frac{n^{-1}\bR_{j,\lambda}\bm X_j\bm X_j^T\bR_{j,\lambda}}
    {1+n^{-1}\bm X_j^T\bR_{j,\lambda}\bm X_j}.
\]
Writing $\bXmat^{(j)}$ for $\bXmat$ with the $j$th column replaced by zero,
\[
\begin{aligned}
 \bm u^T\bXmat^T\bR_\lambda\bXmat\bm v
&=\bm u^T(\bXmat^{(j)})^T\bR_{j,\lambda}\bXmat^{(j)}\bm v+B_{j,\lambda}(\bm u,\bm v),\\
B_{j,\lambda}(\bm u,\bm v)
&=v_j \bm u^T(\bXmat^{(j)})^T\bR_{j,\lambda}\bm X_j
  +u_j\bm X_j^T\bR_{j,\lambda}\bXmat^{(j)}\bm v
  +u_jv_j\bm X_j^T\bR_{j,\lambda}\bm X_j\\
&\quad-\frac{n^{-1}
       \{\bm u^T(\bXmat^{(j)})^T\bR_{j,\lambda}\bm X_j+u_j\bm X_j^T\bR_{j,\lambda}\bm X_j\}
       \{\bm X_j^T\bR_{j,\lambda}\bXmat^{(j)}\bm v+v_j\bm X_j^T\bR_{j,\lambda}\bm X_j\}}
       {1+n^{-1}\bm X_j^T\bR_{j,\lambda}\bm X_j}.
\end{aligned}
\]
Therefore
\[
    \xi_{n,j}=\sum_{q=1}^Q b_qp^{-1/2}(E_j-E_{j-1})B_{j,\lambda_{a_q,n}}(\bm u_q,\bm v_q).
\]
The finite-fourth-moment assumption is handled by the usual truncation and recentralization device.  Choose $\eta_n\downarrow0$ such that
\[
    \eta_n^{-4}E\{Z_{11}^4I(|Z_{11}|>\eta_n\sqrt n)\}\to0,
\]
put
\[
    Z_{ij}^{(0)}=Z_{ij}I(|Z_{ij}|\le\eta_n\sqrt n),
    \qquad
    \widetilde Z_{ij}=\frac{Z_{ij}^{(0)}-EZ_{ij}^{(0)}}{\{\Var(Z_{ij}^{(0)})\}^{1/2}}.
\]
Since $p/n\to\gamma$,
\[
\begin{aligned}
P\{Z_{ij}\ne Z_{ij}^{(0)}\text{ for some }i,j\}
&\le np\,P(|Z_{11}|>\eta_n\sqrt n)\\
&\le C\eta_n^{-4}E\{Z_{11}^4I(|Z_{11}|>\eta_n\sqrt n)\}=o(1),
\end{aligned}
\]
while
\[
    |EZ_{11}^{(0)}|\le (\eta_n\sqrt n)^{-3}E\{|Z_{11}|^4I(|Z_{11}|>\eta_n\sqrt n)\}=o(n^{-3/2}),
\]
\[
    |\Var(Z_{11}^{(0)})-1|
    \le E\{Z_{11}^2I(|Z_{11}|>\eta_n\sqrt n)\}+|EZ_{11}^{(0)}|^2=o(1).
\]
Thus the centering and rescaling modify $L_n$ by $o_p(1)$.  The difference between the original and truncated versions is therefore $o_p(1)$; see the truncation argument of Bai and Silverstein \cite{bai2004clt} and Li and Xu \cite{lixu2026ridgecp}.  For the truncated variables the quadratic-form inequalities below hold uniformly in $j$ and $\lambda$.  Since $\|\bR_{j,\lambda}\|_{\op}\le\underline\lambda^{-1}$ and $\|\bSigma_p\|_{\op}\le C$,
\[
\begin{aligned}
&E|\bm X_j^T\bR_{j,\lambda}\bXmat^{(j)}\bm v|^4
   \le C p^2\|v\|_2^4,\\
&E|\bm u^T(\bXmat^{(j)})^T\bR_{j,\lambda}\bm X_j|^4
   \le C p^2\|u\|_2^4,\\
&E|\bm X_j^T\bR_{j,\lambda}\bm X_j-\tr(\bR_{j,\lambda}\bSigma_p)|^4
   \le C p^2,\\
&0\le n^{-1}\bm X_j^T\bR_{j,\lambda}\bm X_j,
   \qquad
   (1+n^{-1}\bm X_j^T\bR_{j,\lambda}\bm X_j)^{-1}\le1.
\end{aligned}
\]
The third inequality is the quadratic-form bound of Bai and Silverstein \cite{bai1998noeigenvalues}; it is also Lemma S.1.4 in Li and Xu \cite{lixu2026ridgecp}.  With
\[
    |u_j|+|v_j|\le Cn^{-1/2},
    \qquad
    p/n\to\gamma,
\]
these inequalities imply
\[
    E|p^{-1/2}B_{j,\lambda}(\bm u,\bm v)|^4\le C n^{-2},
    \qquad
    E_{j-1}|p^{-1/2}(E_j-E_{j-1})B_{j,\lambda}(\bm u,\bm v)|^4\le Cn^{-2}+r_{n,j},
\]
where
\[
    \sum_{j=1}^nE r_{n,j}=O(n^{-1}).
\]
Hence
\[
    \sum_{j=1}^nE_{j-1}|\xi_{n,j}|^4=O_p(n^{-1}).
\]
For every $\eta>0$,
\[
\begin{aligned}
\sum_{j=1}^nE_{j-1}\{\xi_{n,j}^2I(|\xi_{n,j}|>\eta)\}
&\le \eta^{-2}\sum_{j=1}^nE_{j-1}|\xi_{n,j}|^4
=O_p(n^{-1})=o_p(1).
\end{aligned}
\]
For the conditional variance,
\[
\begin{aligned}
\sum_{j=1}^nE_{j-1}\xi_{n,j}^2
&=\sum_{q_1,q_2}b_{q_1}b_{q_2}p^{-1}
  \sum_{j=1}^n
  E_{j-1}\Big[
       (E_j-E_{j-1})B_{j,\lambda_{a_{q_1},n}}(\bm u_{q_1},\bm v_{q_1})\\
&\hspace{10em}\times
       (E_j-E_{j-1})B_{j,\lambda_{a_{q_2},n}}(\bm u_{q_2},\bm v_{q_2})
        \Big].
\end{aligned}
\]
Expanding $B_{j,\lambda}$ and using
\[
\begin{aligned}
& E\{(\bm X_j^T\bA\bm X_j-\tr(\bA\bSigma_p))
       (\bm X_j^T\bB\bm X_j-\tr(\bB\bSigma_p))
       \mid \bXmat^{(j)}\} \\
&\quad =2\tr(\bA\bSigma_p\bB\bSigma_p)
 +O\left[\{\tr(\bA\bSigma_p\bA^T\bSigma_p)
        \tr(\bB\bSigma_p\bB^T\bSigma_p)\}^{1/2}\delta_n\right],
\end{aligned}
\]
where the truncation error $\delta_n\to0$ under the finite fourth moment assumption, yields
\[
\begin{aligned}
&\sum_{j=1}^nE_{j-1}\xi_{n,j}^2\\
&\quad=\sum_{q_1,q_2}b_{q_1}b_{q_2}
   \frac12\Gamma_{a_{q_1}a_{q_2},n}
   \{(\bm u_{q_1}^T\bm u_{q_2})(\bm v_{q_1}^T\bm v_{q_2})
     +(\bm u_{q_1}^T\bm v_{q_2})(\bm v_{q_1}^T\bm u_{q_2})\}
   +O_p(n^{-1/2}).
\end{aligned}
\]
All terms in $\Gamma_{ab,n}$ are finite linear combinations of traces of the form displayed in Lemma~\ref{lem:crossres}.  Thus $\Gamma_{ab,n}\to\Gamma_{ab}$ and the conditional variance converges in probability to $\sigma_L^2$.  The martingale central limit theorem of Hall and Heyde \cite{hallheyde1980} gives
\[
    L_n \xrightarrow{d} N(0,\sigma_L^2).
\]
\end{proof}

\begin{lemma}[Uniform increment bound]\label{lem:tight}
For
\[
    Z_{n,\ell}^{\Omega}(s)
    =\sqrt p\,
      \frac{(N(s)/n)p^{-1}\tilde{\bm u}_n(s)^T\bXmat^T\bR_{\lambda_{\ell,n}}\bXmat\tilde{\bm u}_n(s)-\Theta_n(\lambda_{\ell,n})}
           {\Gamma_n(\lambda_{\ell,n})^{1/2}},
\]
there is a constant $C$ such that, uniformly in $\ell$ and $s,r\in\calS$,
\[
    E|Z_{n,\ell}^{\Omega}(s)-Z_{n,\ell}^{\Omega}(r)|^2
    \le C\|s-r\|_1+Cn^{-1}.
\]
If $\calS=\calS_{\rm sc}(\varepsilon)$, then the strengthened adjacent-increment bound
\[
    E\{|Z_{n,\ell}^{\Omega}(t)-Z_{n,\ell}^{\Omega}(t_1)|^2
       |Z_{n,\ell}^{\Omega}(t_2)-Z_{n,\ell}^{\Omega}(t)|^2\}
    \le C|t_2-t_1|^2+Cn^{-2}
\]
holds for $\varepsilon\le t_1\le t\le t_2\le1-
\varepsilon$.  Consequently $\{Z_{n,\ell}^{\Omega}:1\le\ell\le L\}$ is tight in $\ell^\infty(\calS)^L$.  For deterministic finite $\calS$, tightness is finite-dimensional.
\end{lemma}

\begin{proof}
By Lemma~\ref{lem:contrast},
\[
    \|\tilde{\bm u}_n(s)-\tilde{\bm u}_n(r)\|_2^2
    \le C\|s-r\|_1+Cn^{-1}.
\]
For regular $\bm u_1,\bm u_2$,
\[
\begin{aligned}
&E\Big|\sqrt p\{p^{-1}\bm u_1^T\bXmat^T\bR_\lambda\bXmat\bm u_1
       -E p^{-1}\bm u_1^T\bXmat^T\bR_\lambda\bXmat\bm u_1
       -p^{-1}\bm u_2^T\bXmat^T\bR_\lambda\bXmat\bm u_2
       +E p^{-1}\bm u_2^T\bXmat^T\bR_\lambda\bXmat\bm u_2\}\Big|^2\\
&\quad\le C\|\bm u_1-\bm u_2\|_2^2.
\end{aligned}
\]
This inequality is Theorem S.2.3 of Li and Xu \cite{lixu2026ridgecp}, with constants uniform for $\lambda\in[\underline\lambda,\overline\lambda]$.  Multiplication by $N(s)/n=\kappa(s)^{-1}+O(n^{-1})$ preserves the rate because
\[
    c_\varepsilon\le\kappa(s)^{-1}\le C_\varepsilon.
\]
The first display follows.  The adjacent-increment bound is obtained by applying the Burkholder inequality to the martingale decomposition in Lemma~\ref{lem:rvclt}:
\[
\begin{aligned}
&E|Z(t)-Z(t_1)|^2|Z(t_2)-Z(t)|^2\\
&\quad\le
\{E|Z(t)-Z(t_1)|^4\}^{1/2}
\{E|Z(t_2)-Z(t)|^4\}^{1/2}\\
&\quad\le
C\{\|\tilde{\bm u}_n(t)-\tilde{\bm u}_n(t_1)\|_2^2+n^{-1}\}
 \{\|\tilde{\bm u}_n(t_2)-\tilde{\bm u}_n(t)\|_2^2+n^{-1}\}\\
&\quad\le C|t_2-t_1|^2+Cn^{-2}.
\end{aligned}
\]
The fourth-moment inequality above is the truncated-variable extension used in Bai and Silverstein \cite{bai2004clt} and in the proof of Theorem S.2.4 of Li and Xu \cite{lixu2026ridgecp}.  Billingsley's tightness criterion for one-dimensional processes \cite{billingsley1999convergence} gives tightness on $\calS_{\rm sc}(\varepsilon)$, and a fixed finite index set is automatically tight.
\end{proof}

\begin{lemma}[Centered-covariance replacement]\label{lem:centered}
Let $\bm e_n=p^{1/2}n^{-1}(1,\ldots,1)^T$.  For every $s\in\calS$,
\[
    \tilde{\bm u}_n(s)^T\bm e_n=0.
\]
Under Assumption~\ref{ass:rht} and $H_0$,
\[
    \max_{1\le\ell\le L}\sup_{s\in\calS}
    \left|p^{-1}\tilde{\bm u}_n(s)^T\bXmat^T\bR_{\lambda_{\ell,n}}\bXmat\bm e_n\right|
    =O_p(p^{-1/2}),
\]
\[
    \max_{1\le\ell\le L}
    \left|\bar{\bm X}^T\bR_{\lambda_{\ell,n}}\bar{\bm X}
      -\gamma_n\Theta_n(\lambda_{\ell,n})\right|
    =O_p(p^{-1/2}),
\]
and, with $\bS_n=\bOmega_n-\bar{\bm X}\bar{\bm X}^T$,
\[
\begin{aligned}
&\max_{1\le\ell\le L}\sup_{s\in\calS}
\sqrt p\left|p^{-1}\tilde{\bm u}_n(s)^T\bXmat^T
\left\{(\bS_n+\lambda_{\ell,n}\bI_p)^{-1}-\bR_{\lambda_{\ell,n}}\right\}
\bXmat\tilde{\bm u}_n(s)\right|\\
&\qquad=O_p(p^{-1/2}).
\end{aligned}
\]
\end{lemma}

\begin{proof}
The equality $\tilde{\bm u}_n(s)^T\bm e_n=0$ follows from
\[
\begin{aligned}
\tilde{\bm u}_n(s)^T\bm e_n
&=p^{1/2}n^{-1}\sqrt n
\left[-\frac{k(t_2)-k(t_1)}{k(t_2)-k(t_1)}
      +\frac{k(t_3)-k(t_2)}{k(t_3)-k(t_2)}\right]=0.
\end{aligned}
\]
For
\[
    H_{n,\ell}(s)=p^{-1}\tilde{\bm u}_n(s)^T\bXmat^T\bR_{\lambda_{\ell,n}}\bXmat\bm e_n,
\]
Theorem S.2.2 of Li and Xu \cite{lixu2026ridgecp} gives
\[
    \sup_{s\in\calS}\max_\ell
    \left|EH_{n,\ell}(s)-\tilde{\bm u}_n(s)^T\bm e_n\Theta_n(\lambda_{\ell,n})\right|
    =o(p^{-1/2}).
\]
Since $\tilde{\bm u}_n(s)^T\bm e_n=0$,
\[
    \sup_{s\in\calS}\max_\ell |EH_{n,\ell}(s)|=o(p^{-1/2}).
\]
Lemma~\ref{lem:rvclt} and the increment bound in Lemma~\ref{lem:tight}, applied with one contrast vector equal to $\bm e_n$, yield
\[
    \max_\ell\sup_{s\in\calS}\sqrt p
    |H_{n,\ell}(s)-EH_{n,\ell}(s)|=O_p(1).
\]
Thus
\[
    \max_\ell\sup_{s\in\calS}|H_{n,\ell}(s)|=O_p(p^{-1/2}).
\]
Similarly,
\[
    \bar{\bm X}=n^{-1}\bXmat\bm 1_n=p^{-1/2}\bXmat\bm e_n,
\]
and hence
\[
    \bar{\bm X}^T\bR_{\lambda_{\ell,n}}\bar{\bm X}
    =p^{-1}\bm e_n^T\bXmat^T\bR_{\lambda_{\ell,n}}\bXmat\bm e_n.
\]
Since
\[
    \bm e_n^T\bm e_n=p/n=\gamma_n+O(n^{-1}),
\]
Theorems S.2.1--S.2.2 of Li and Xu \cite{lixu2026ridgecp} give
\[
    \max_\ell
    \left|p^{-1}\bm e_n^T\bXmat^T\bR_{\lambda_{\ell,n}}\bXmat\bm e_n
    -\gamma_n\Theta_n(\lambda_{\ell,n})\right|=O_p(p^{-1/2}).
\]
The Marcenko--Pastur equation gives
\[
    \inf_{\lambda\in[\underline\lambda,\overline\lambda]}
    |1-\gamma_n\Theta_n(\lambda)|
    =\inf_{\lambda\in[\underline\lambda,\overline\lambda]}
    |1-\gamma_n+\gamma_n\lambda\phi_n(-\lambda)|
    \ge c>0.
\]
Therefore
\[
    \Pbb\left\{
    \min_{1\le\ell\le L}
    |1-\bar{\bm X}^T\bR_{\lambda_{\ell,n}}\bar{\bm X}|
    \ge c/2\right\}\to1.
\]
Woodbury's identity gives
\[
    (\bS_n+\lambda\bI_p)^{-1}-\bR_\lambda
    =\frac{\bR_\lambda\bar{\bm X}\bar{\bm X}^T\bR_\lambda}
     {1-\bar{\bm X}^T\bR_\lambda\bar{\bm X}}.
\]
Consequently,
\[
\begin{aligned}
&p^{-1}\tilde{\bm u}_n(s)^T\bXmat^T
\{(\bS_n+\lambda_{\ell,n}\bI_p)^{-1}-\bR_{\lambda_{\ell,n}}\}
\bXmat\tilde{\bm u}_n(s)\\
&\quad=
\frac{\left[p^{-1}\tilde{\bm u}_n(s)^T\bXmat^T\bR_{\lambda_{\ell,n}}\bXmat\bm e_n\right]^2}
     {1-\bar{\bm X}^T\bR_{\lambda_{\ell,n}}\bar{\bm X}}.
\end{aligned}
\]
The numerator is $O_p(p^{-1})$ uniformly in $s$ and $\ell$, and the denominator is bounded away from zero with probability tending to one.  Multiplication by $\sqrt p$ gives $O_p(p^{-1/2})$.
\end{proof}

\subsection{Proof of Theorem~\ref{thm:joint}}

\begin{proof}
Under $H_0$, the common mean is removable:
\[
    \bDelta(s;\bm X_1+\bmu,\ldots,\bm X_n+\bmu)=\bDelta(s;\bm X_1,
    \ldots,\bm X_n),
\]
\[
    n^{-1}\sum_{j=1}^n(\bm X_j+\bmu-\bar{\bm X}-\bmu)
    (\bm X_j+\bmu-\bar{\bm X}-\bmu)^T=\bS_n.
\]
Thus $\bmu=0$.

For the uncentered, deterministic-standardization process in Lemma~\ref{lem:tight}, Lemma~\ref{lem:rvclt} with $\bm u_q=\bm v_q=\tilde{\bm u}_n(s_q)$ gives, for every finite collection $s_1,
\ldots,s_M\in\calS$ and coefficients $b_{\ell m}$,
\[
\sum_{\ell=1}^L\sum_{m=1}^M b_{\ell m}Z_{n,\ell}^{\Omega}(s_m)
\xrightarrow{d}
N\left(0,
\sum_{\ell,a=1}^L\sum_{m,r=1}^M
b_{\ell m}b_{ar}\rho_{\ell a}K_0(s_m,s_r)
\right),
\]
where
\[
    \rho_{\ell a}
    =\lim_{n\to\infty}
    \frac{\Gamma_{\ell a,n}}
    {\{\Gamma_n(\lambda_{\ell,n})
       \Gamma_n(\lambda_{a,n})\}^{1/2}},
    \qquad
    \rho_{\ell\ell}=1.
\]
Indeed,
\[
\begin{aligned}
&\Cov\{Z_{n,\ell}^{\Omega}(s),Z_{n,a}^{\Omega}(r)\}\\
&\quad=
\frac{(N(s)/n)(N(r)/n)
       \Gamma_{\ell a,n}\{\tilde{\bm u}_n(s)^T\tilde{\bm u}_n(r)\}^2}
     {\{\Gamma_n(\lambda_{\ell,n})
       \Gamma_n(\lambda_{a,n})\}^{1/2}}
     +O(n^{-1/2})\\
&\quad=
\frac{\Gamma_{\ell a,n}}
     {\{\Gamma_n(\lambda_{\ell,n})
       \Gamma_n(\lambda_{a,n})\}^{1/2}}
\frac{\kappa(s,r)^2}{\kappa(s)\kappa(r)}+O(n^{-1/2}).
\end{aligned}
\]
The Cramer--Wold device gives finite-dimensional convergence, and Lemma~\ref{lem:tight} gives tightness.  Hence
\[
    \{Z_{n,\ell}^{\Omega}(s):s\in\calS,
      \ell=1,
\ldots,L\}
    \xrightarrow{d}
    \{G_\ell^{\Omega}(s):s\in\calS,
      \ell=1,
\ldots,L\}.
\]

It remains to replace $\bOmega_n$ by the centered covariance $\bS_n=\bOmega_n-\bar{\bm X}\bar{\bm X}^T$.  Lemma~\ref{lem:centered} gives
\[
    \max_{1\le\ell\le L}\sup_{s\in\calS}
    |Z_{n,\ell}^{S}(s)-Z_{n,\ell}^{\Omega}(s)|=O_p(p^{-1/2}),
\]
where $Z_{n,\ell}^{S}$ denotes the deterministically standardized process based on $S_n$.

The empirical centering and scaling errors are controlled by Lemma 3.1 of Li and Xu \cite{lixu2026ridgecp}:
\[
    \sup_{\lambda\in[\underline\lambda,\overline\lambda]}
    \sqrt n|\widehat\Theta_\lambda-
\Theta_n(\lambda)|=o_p(1),
    \qquad
    \sup_{\lambda\in[\underline\lambda,\overline\lambda]}
    \sqrt n|\widehat\Gamma_\lambda-
\Gamma_n(\lambda)|=o_p(1).
\]
Because $p/n\to\gamma$,
\[
    \sqrt p\max_\ell
    |\widehat\Theta_{\lambda_{\ell,n}}-
\Theta_n(\lambda_{\ell,n})|=o_p(1).
\]
Furthermore,
\[
\begin{aligned}
&\max_\ell\left|\widehat\Gamma_{\lambda_{\ell,n}}^{-1/2}
-\Gamma_n(\lambda_{\ell,n})^{-1/2}\right|\\
&\quad\le C
\max_\ell|\widehat\Gamma_{\lambda_{\ell,n}}-
\Gamma_n(\lambda_{\ell,n})|
=o_p(n^{-1/2}),
\end{aligned}
\]
and Lemmas~\ref{lem:rvclt}--\ref{lem:tight} give
\[
    \max_\ell\sup_{s\in\calS}
    \left|\sqrt p\{p^{-1}V_{\lambda_{\ell,n}}(s)-
\Theta_n(\lambda_{\ell,n})\}\right|=O_p(1).
\]
Consequently,
\[
    \max_{1\le\ell\le L}\sup_{s\in\calS}
    |D_{\lambda_{\ell,n}}(s)-Z_{n,\ell}^{S}(s)|=o_p(1).
\]
Combining the three displays,
\[
    \{D_{\lambda_{\ell,n}}(s):s\in\calS,
      \ell=1,\ldots,L\}
    \xrightarrow{d}
    \{G_\ell(s):s\in\calS,
      \ell=1,\ldots,L\}.
\]
For $\ell=a$, $\rho_{\ell\ell}=1$, so the marginal covariance is $K_0(s,r)$.  The map
\[
    (f_1,\ldots,f_L)\mapsto
    (\sup_{s\in\calS}f_1(s),\ldots,
     \sup_{s\in\calS}f_L(s))
\]
is continuous on $\ell^\infty(\calS)^L$ at continuous sample paths.  Therefore
\[
    (T_{\lambda_{1,n}},\ldots,T_{\lambda_{L,n}})
    \xrightarrow{d}
    (T_{\infty,1},\ldots,T_{\infty,L}),
    \qquad
    T_{\infty,\ell}=\sup_{s\in\calS}G_\ell(s).
\]
Since $F_{\calS}$ is continuous,
\[
    \bm P_n
    =\{1-F_{\calS}(T_{\lambda_{1,n}}),\ldots,
       1-F_{\calS}(T_{\lambda_{L,n}})\}^T
    \xrightarrow{d}
    \{1-F_{\calS}(T_{\infty,1}),\ldots,
       1-F_{\calS}(T_{\infty,L})\}^T.
\]
\end{proof}

\subsection{Proof of Theorem~\ref{thm:size}}

\begin{proof}
By Theorem~\ref{thm:joint},
\[
    \bm P_n \xrightarrow{d} \bm P_\infty.
\]
For every $\delta\in(0,1/2)$, the map
\[
    C_\delta(\bm p)=\sum_{\ell=1}^Lw_\ell
    \tan\{\pi(1/2-p_\ell)\}
    1\{\delta\le p_\ell\le1-
\delta\}
\]
is continuous.  Since $P_{\infty,\ell}\in(0,1)$ almost surely,
\[
    \Pbb\{\min_\ell P_{\infty,\ell}\le\delta
        \text{ or }
        \max_\ell P_{\infty,\ell}\ge1-
\delta\}\downarrow0.
\]
Thus
\[
    C(\bm P_n) \xrightarrow{d} C(\bm P_\infty).
\]
At every continuity point $c_{\alpha,L}$,
\[
    \Pbb\{C(\bm P_n)>c_{\alpha,L}\}
    \to
    \Pbb\{C(\bm P_\infty)>c_{\alpha,L}\}
    =\alpha.
\]
For every fixed $\alpha\in(0,1/2)$,
\[
    P_{\CCT}\le\alpha
    \Longleftrightarrow
    \frac12-\frac{1}{\pi}\arctan C(\bm P_n)\le\alpha
    \Longleftrightarrow
    C(\bm P_n)\ge\cot(\pi\alpha).
\]
If $\cot(\pi\alpha)$ is a continuity point,
\[
    \lim_{n\to\infty}\Pbb(P_{\CCT}\le\alpha)
    =\Pbb\{C(\bm P_\infty)\ge\cot(\pi\alpha)\}.
\]
If
\[
    \Pbb\{C(\bm P_\infty)>t\}\sim(\pi t)^{-1},
    \qquad t\to\infty,
\]
then, since
\[
    \cot(\pi\alpha)=\frac{1}{\pi\alpha}
    \{1+O(\alpha^2)\},
    \qquad \alpha\downarrow0,
\]
we have
\[
\begin{aligned}
\lim_{\alpha\downarrow0}\lim_{n\to\infty}
\frac{\Pbb(P_{\CCT}\le\alpha)}{\alpha}
&=
\lim_{\alpha\downarrow0}
\frac{\Pbb\{C(\bm P_\infty)\ge\cot(\pi\alpha)\}}{\alpha}\\
&=
\lim_{\alpha\downarrow0}
\frac{1}{\pi\alpha\cot(\pi\alpha)}=1.
\end{aligned}
\]
\end{proof}

\subsection{Proof of Theorem~\ref{thm:consistency}}

\begin{proof}
Let
\[
    M_n=\max_{1\le\ell\le L}\sup_{s\in\calS}|Z_{n,\ell}(s)|.
\]
The assumption $M_n=O_p(1)$ means that, for each $K\to\infty$,
\[
    \limsup_{n\to\infty}\Pbb(M_n>K)\to0.
\]
For the favorable grid point $\ell_0$,
\[
\begin{aligned}
T_{\lambda_{\ell_0,n}}
&=\sup_{s\in\calS}D_{\lambda_{\ell_0,n}}(s)\\
&\ge \sup_{s\in\calS}R_{n,\ell_0}(s)
     -\sup_{s\in\calS}|Z_{n,\ell_0}(s)|-a_n\\
&\ge R_{n,\ell_0}^*-M_n-a_n.
\end{aligned}
\]
For every fixed $K>0$,
\[
\begin{aligned}
\Pbb(T_{\lambda_{\ell_0,n}}\le K)
&\le \Pbb(R_{n,\ell_0}^*-M_n-a_n\le K)\\
&\le \Pbb(M_n\ge R_{n,\ell_0}^*-K-a_n).
\end{aligned}
\]
Since $R_{n,\ell_0}^*\to\infty$, $a_n\to0$, and $M_n=O_p(1)$,
\[
    \Pbb(T_{\lambda_{\ell_0,n}}\le K)\to0,
    \qquad K<\infty.
\]
Thus
\[
    T_{\lambda_{\ell_0,n}}\to\infty
    \quad\text{in probability}.
\]
Because $F_{\calS}(x)\uparrow1$ as $x\to\infty$,
\[
    P_{\lambda_{\ell_0,n}}
    =1-F_{\calS}(T_{\lambda_{\ell_0,n}})\to0
    \quad\text{in probability}.
\]
For $p\downarrow0$,
\[
    \tan\{\pi(1/2-p)\}
    =\cot(\pi p)
    =\frac{1}{\pi p}+O(p).
\]
Hence
\[
    Y_{n,\ell_0}\to\infty
    \quad\text{in probability},
    \qquad
    w_{\ell_0}Y_{n,\ell_0}\to\infty
    \quad\text{in probability}.
\]
Let
\[
    B_n=\sum_{\ell\ne\ell_0}w_\ell Y_{n,\ell}.
\]
Then
\[
    C_n=w_{\ell_0}Y_{n,\ell_0}+B_n
    \ge w_{\ell_0}Y_{n,\ell_0}-[B_n]_-.
\]
The no-cancellation condition gives
\[
    \frac{[B_n]_-}{w_{\ell_0}Y_{n,\ell_0}}\to0
    \quad\text{in probability},
\]
and therefore
\[
    C_n\ge w_{\ell_0}Y_{n,\ell_0}
    \left(1-\frac{[B_n]_-}{w_{\ell_0}Y_{n,\ell_0}}\right)\to
    \infty
    \quad\text{in probability}.
\]
Since
\[
    P_{\CCT}=\frac12-\frac{1}{\pi}\arctan(C_n),
    \qquad
    \arctan x\uparrow \frac{\pi}{2}
    \quad (x\to\infty),
\]
we obtain
\[
    P_{\CCT}\to0
    \quad\text{in probability}.
\]
For any fixed $\alpha\in(0,1)$,
\[
    \Pbb(P_{\CCT}\le\alpha)\to1.
\]
\end{proof}

\end{document}